\documentclass[aps,pra,twocolumn,superscriptaddress,10pt,article,nofootinbib]{revtex4}

\usepackage[ascii]{inputenc}
\usepackage{dsfont}
\usepackage[bookmarksnumbered]{hyperref}
\usepackage{breakurl}
\usepackage{amsfonts}
\usepackage{amsmath}
\usepackage{amsthm}
\usepackage{amssymb}
\usepackage{mathtools}
\usepackage{stmaryrd}
\usepackage{braket}
\usepackage{color}
\usepackage{aliascnt}
\usepackage{tikz}
\usetikzlibrary{backgrounds,fit,calc,matrix,patterns,decorations.pathreplacing,shapes.geometric,spy,shapes.symbols,shadows}
\tikzstyle{ball} = [circle,shading=ball, ball color=blue!50!white,  minimum size=1cm]
\tikzstyle{blank} = [fill=white,minimum size=0em]
\tikzstyle{seethrough} = [minimum size=0em]
\tikzstyle{tensor} = [draw,fill=white,minimum height=1.8em,minimum width=2.5em,inner sep=0pt] 
\tikzstyle{phase} = [draw,fill,shape=circle,minimum size=5pt,inner sep=0pt]
\tikzstyle{state} = [draw,fill=white,minimum size=0.5em,inner sep=0pt]
\tikzstyle{unitary} = [draw,fill=white,minimum height=1em,minimum width=3em,inner sep=0.1pt] 
\usepackage{url}
\usepackage{times}
\usepackage{chngpage}
\usepackage{framed}

\DeclareMathOperator{\poly}{poly}
\newcommand{\newjointcountertheorem}[3]{
  \newaliascnt{#1}{#2}
  \newtheorem{#1}[#1]{#3}
  \aliascntresetthe{#1}
}
\newtheorem*{thm-plain}{Theorem}
\newtheorem{thm}{Theorem}
\newjointcountertheorem{prp}{thm}{Proposition}
\newjointcountertheorem{lem}{thm}{Lemma}
\newjointcountertheorem{cor}{thm}{Corollary}
\newjointcountertheorem{cnj}{thm}{Conjecture}
\newjointcountertheorem{defn}{thm}{Definition}
\newjointcountertheorem{remark}{thm}{Remark}
\newjointcountertheorem{example}{thm}{Example}
\def\Snospace~{\S{}}

\usetikzlibrary{arrows,shapes,positioning}
\usetikzlibrary{decorations.markings}
\tikzstyle arrowstyle=[scale=1]
\tikzstyle directed=[postaction={decorate,decoration={markings,
    mark=at position .7 with {\arrow[arrowstyle]{stealth};}}}]
\tikzstyle reverse directed=[postaction={decorate,decoration={markings,
    mark=at position .7 with {\arrowreversed[arrowstyle]{stealth};}}}]

\allowdisplaybreaks[4]

\newcommand{\half}{\frac{1}{2}}
\newcommand{\restr}{{
  \left.\kern-\nulldelimiterspace 
  \vphantom{\big|} 
  \right| 
  }}

\begin{document}

\title{Symmetry-protected adiabatic quantum transistors}

\author{Dominic J. Williamson}
\affiliation{Centre for Engineered Quantum Systems, School of Physics, The University of Sydney, Sydney, NSW 2006, Australia}
\affiliation{Vienna Center for Quantum Science and Technology, Faculty of Physics, \\University of Vienna, A-1090 Wien, Austria}
\author{Stephen D. Bartlett}
\affiliation{Centre for Engineered Quantum Systems, School of Physics, The University of Sydney, Sydney, NSW 2006, Australia}
\date{13 February 2015} 

\begin{abstract}
Adiabatic quantum transistors allow quantum logic gates to be performed by applying a large field to a quantum many-body system prepared in its ground state, without the need for local control.  The basic operation of such a device can be viewed as driving a spin chain from a symmetry-protected phase to a trivial phase. This perspective offers an avenue to generalise the adiabatic quantum transistor and to design several improvements.  The performance of quantum logic gates is shown to depend only on universal symmetry properties of a symmetry-protected phase rather than any fine tuning of the Hamiltonian, and it is possible to implement a universal set of logic gates in this way by combining several different types of symmetry protected matter.  Such symmetry-protected adiabatic quantum transistors are argued to be robust to a range of relevant noise processes.
\end{abstract}
  
\maketitle

\section{Introduction}

Quantum computers promise a computational speedup for problems believed to be hard to solve using classical computers. There are many different architectures for the implementation of quantum computation (QC), each realising the same computational power with different requirements on physical hardware. Along with the canonical quantum circuit model, there is measurement-based QC \cite{briegel2001persistent}, adiabatic QC \cite{aharonov200445th}, holonomic QC \cite{zanardi1999holonomic}, and topological QC \cite{kitaev1997quantum}, as well as variations that combine aspects from different models.  One such hybrid scheme --- the adiabatic quantum transistor (AQT) model proposed by Bacon, Flammia and Crosswhite~\cite{bacon2012adiabatic} --- is appealing from a practical perspective, as it demands very minimal control requirements.  This model is technically open loop holonomic QC but draws upon aspects of all the aforementioned architectures, and requires only the ability to prepare the ground state of an interacting many-body Hamiltonian and perform an adiabatic application of a global control field, without the need for any local control.

The AQT model builds upon earlier work on one-dimensional quantum computational wires and their usefulness for measurement-based and holonomic quantum computation~\cite{gross2007measurement,bacon2009adiabatic,bacon2010adiabatic,bartlett2010,kaltenbaek2010,miyake2010quantum,renes2011holonomic,else2012symmetry}. These models are best understood in terms of computation on information encoded in the correlations amongst qubits in the ground state~\cite{gross2007measurement}, which can also be viewed from a very recent perspective as fractionalized edge modes associated with the boundaries of symmetry-protected phases of spin chains~\cite{miyake2010quantum,renes2011holonomic}. In the measurement-based model, a very precise relationship between the computational properties of a spin chain and its symmetry-protected quantum order was developed by Else \textit{et al.}~\cite{else2012symmetry,else2012symmetry2}. Such a precise relationship is lacking in the open loop holonomic QC setting, and so the general physical principles that define the AQT models and give rise to their special properties are not yet explored.

In this paper, we show that the operation of an adiabatic quantum transistor can be viewed as driving the system through a symmetric phase transition, from a symmetry-protected (SP) phase to a trivial symmetric phase, using a global control field.  Within this perspective, we extend the specific AQT gates defined using finely tuned model Hamiltonians in Ref.~\cite{bacon2012adiabatic} to whole SP phases of matter, thereby further reducing the control requirements for this scheme.  Such adiabatic quantum transistors that are based solely on the properties of symmetry-protected phases can be called \emph{symmetry-protected adiabatic quantum transistors} (SPAQT). 

This new perspective in terms of processing quantum logic at the boundary of a SP ordered phase provides several other natural generalizations of the AQT model.  We explore the degenerate ground state encoding used in the AQT model in terms of the defining properties of a SP phase, and in doing so, determine the quantum logic gates that can be implemented by a spin chain referencing only the symmetries of a SP phase. We show how multiple different logic gates can be performed by preserving distinct subgroups of a larger symmetry group during the evolution. Finally we address the issue of errors within the model, drawing the distinction between errors to which the model is inherently robust and errors which will require standard fault tolerance constructions.

The paper is laid out as follows.  In Sec.~\ref{sec:SPintro} we begin with a brief review of symmetry-protected (topological) phases of matter, followed by Sec.~\ref{sec:SPencoding} with an explanation of the general ground state encoding used for any SP phase and its robustness properties throughout the phase.  In Sec.~\ref{sec:ElementaryGate}, we present the general process to implement an elementary logic gate upon the encoded information by adiabatically shifting a phase boundary between SP and trivial matter by a single lattice spacing.  We go on to describe the generic requirements needed to achieve a universal gate set with SP chains in Sec.~\ref{sec:GeneralGates}.  Then, building upon the basic gate construction, we explain in Sec.~\ref{sec:Transistor} how one can implement a symmetry-protected adiabatic quantum transistor gate by adiabatically traversing a symmetric phase transition from a SP to trivial phase.  We discuss the robustness of our proposed scheme to a large class of realistic errors in Sec.~\ref{sec:Errors}, and conclude with a discussion of these results and future directions in Sec.~\ref{sec:Conclusions}.  Explicit details of the operation of an SPAQT based on the Haldane phase of a spin-1 chain are presented in the Appendix, building on the results of Ref.~\cite{renes2011holonomic} and offering several new results.

\section{Symmetry-protected phases}
\label{sec:SPintro}

In this section, we review the definition of SP phases in one-dimension, and their characterisation in terms of the second cohomology class of a symmetry group.  For further details, see Refs.~\cite{chen2011classification,schuch2011classifying,else2012symmetry}.
 
A zero temperature quantum phase is defined as a family of uniformly gapped Hamiltonians (and their ground states) on periodic, regular lattices of all finite sizes that are equivalent under constant-length, adiabatic evolutions that preserve a uniform gap.  A richer set of phases arise when considering Hamiltonians that commute with a given representation of a symmetry group.  Allowing only adiabatic paths that also commute with the symmetry, a SP phase is defined to be a class of symmetric uniformly gapped Hamiltonians $H$ that are equivalent under such symmetry-respecting adiabatic evolutions.  We will consider only on-site symmetries, i.e., those whose representations take the form of a tensor product of the same representation $U_g$ on each physical site, and hence the symmetry condition is $\left[U_g^{\otimes N},H\right]=0, \ \forall g\in G$.  (We note that symmetries which are not on-site may also support SP phases, but we do not consider these here.)  We further restrict our consideration to only those symmetric Hamiltonians possessing a unique symmetric ground state under periodic boundary conditions. In this setting distinct equivalence classes emerge, the class containing a symmetric product state is a trivial symmetric phase and other distinct classes are called symmetry-protected (SP) phases.

A consequence of this definition is that no constant length, gapped, symmetric adiabatic evolution of a ground state in a non-trivial SP phase can map it to a product state. This can be interpreted as being due to a non-trivial symmetry-protected entanglement structure that persists in all ground states of a SP phase \cite{chen2010local}. 
In one dimension, we identify this entanglement by a Schmidt rank greater than one across any bipartition of the system. This entanglement can be intuitively understood as arising from a pair of maximally entangled, fractionalised virtual particles localised on opposite sides of the bipartition which cannot be disentangled by any symmetric adiabatic evolution. There may also be some trivial entanglement present across the bipartition that can be removed by symmetric adiabatic evolution and hence is not robust throughout the phase and consequently not symmetry-protected. Only the nontrivial, symmetry-protected entanglement persists in a renormalisation fixed point state of a SP phase. A chain with open boundary conditions can be viewed as a periodic chain that has been cut open and had its boundaries separated.  In this case, the virtual particles at the boundaries are not correlated, as they are separated by an arbitrary distance of gapped bulk material with exponentially decaying correlations, and because we have removed the relevant Hamiltonian term that was coupling them, they become free.  Hence a nontrivial SP Hamiltonian on an open chain will possess some ground state degeneracy that can be associated to these fractionalised edge modes.

Because the Hamiltonian commutes with the symmetry, its degenerate ground space is closed under the symmetry action. We can therefore restrict the representation of the symmetry group to the ground space. Because the symmetry group acts locally in tensor product on the chain it cannot entangle the two spatially separated edge modes and therefore acts as a tensor product on the left and right edges modes, as $V^L_g\otimes V^R_g$.  While $V^L_g\otimes V^R_g$ must form a unitary representation of the group, there is additional freedom in defining the individual representations $V_g^{L/R}$.  In particular, we can allow an equal and opposite phase in the multiplication rules of the group action on left and right modes, i.e., 
\begin{equation}
  V_g^{L/R} V_h^{L/R} = \omega(g,h)_{L/R} V_{gh}^{L/R}\,,
\end{equation} 
where $\omega(g,h)_{L/R}$ is a function on $G \times G$ giving a $U(1)$ phase subject to $\omega(g,h)_L\omega(g,h)_R$ being trivial (in the second cohomology sense explained below).  The associativity of the multiplication enforces a constraint on the $\omega$ phases.  We further take equivalence classes under multiplication of each $V_g$ by some arbitrary phase function $\beta(g)$ on $G$. The resulting equivalence classes of phase functions $\omega$ on $G \times G$ form the second cohomology group of the symmetry $H^2(G,U(1))$ \cite{chen2013symmetry}. These equivalence classes are in one-to-one correspondence with the different SP phases possible in one dimension \cite{chen2011classification,schuch2011classifying}, and so provide a way of labeling the different SP phases.  Given the cohomology label of the phase the edge modes must transform under some projective representation $V_g$ of the group with the specified cohomology.

\section{Symmetry-protected ground state encoding}
\label{sec:SPencoding}

In this section, we will describe how quantum information can be encoded into the fractionalized degree of freedom associated with a single edge mode of a non-trivial SP phase of a one-dimensional spin chain.  

It is well known that topologically ordered phases of matter provide families of quantum codes that are insensitive to the microscopic detail of the Hamiltonian. Symmetry-protected phases provide degenerate ground states with similar robustness but only to errors that obey a certain symmetry condition. The defining symmetries of a SP phase also provide uniform global operations that enact logical transformations upon the ground space. Remarkably both the encoding and global logical operations are robust to local perturbations of the parent Hamiltonian, so long as these perturbations are symmetric. 

The essential property of a SP phase that enables us to encode information into the ground space is an equivalence of the global symmetry operators and a projective representation acting upon the edge mode.  In many ways, these global symmetry operators are analogous to logical operators of a stabilizer code, as they commute with the constraints (Hamiltonian terms) that define the code and have a nontrivial action upon the encoded information. 
In the particular example of the cluster state model used to define the AQT of Bacon \textit{et al.}~\cite{bacon2010adiabatic}, which has a stabilizer parent Hamiltonian, these symmetry operators are logical operators.  In the general case, the terms in the Hamiltonian whose ground space is the `code space' do not necessarily commute with one another.

The physical systems we consider for the remainder of the paper are chains of spin degrees of freedom, with interactions governed by a spatially local Hamiltonian 
\begin{equation}
\label{eq:1dlh}
H_N=\sum_{s=0}^{N-1} H_s
\end{equation}
where the $H_s$ terms act on a constant number of spins around site $s$ and the energy scale is normalized such that $\| H_s \| \leq 1$. We only consider models where the energy gap $\Delta$ between the smallest set of quasi-degenerate eigenvalues (meaning their energy spacing shrinks exponentially as the size of the chain grows) and the next lowest eigenvalue (the first excitation) is uniformly lower bounded by a constant as $N$ increases. We restrict our attention to models where the Hamiltonian's terms all commute with a representation of some symmetry group $G$. The representations we will consider are $N$-fold tensor products $U_g^{\otimes N}$ of some on-site unitary representations $U_g$, and so the relevant symmetry condition is $[ H_N , U_g^{\otimes N}] = 0$ for all $g\in G$.

\subsection{Isolating one edge mode}

\begin{figure}[t]
\center
\includegraphics[width=0.95\linewidth]{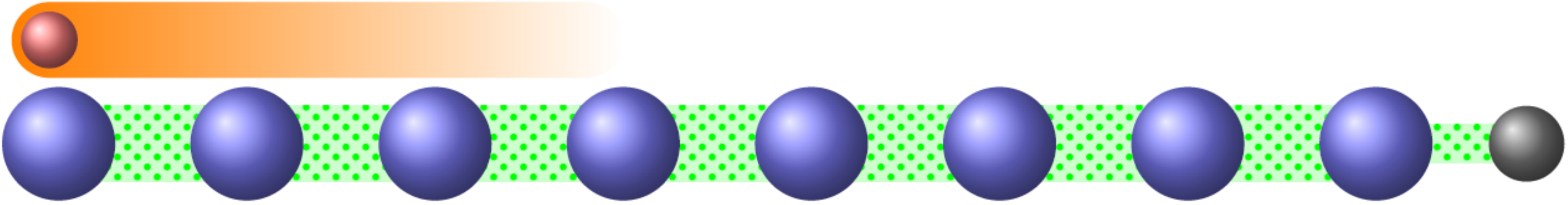}
\caption{A ground state of a symmetric Hamiltonian with one boundary condition (right) fixed.  Large (blue) spheres denote the spins of the chain, each of which transform as $U_g$ under the symmetry action.  A fractional particle, denoted on the right by a small (grey) sphere, transforms as $V_g^{\rm edge}$.  The left edge carries a fractionalized edge degree of freedom (orange), transforming as $U_g^{\otimes N}\otimes V_g^{\text{phys}} |_{E_0} \sim V_g$.}\label{fig:edge}
\end{figure}

A one-dimensional spin chain with fixed boundary conditions from a non-trivial SP phase provides two fractionalized edge modes -- one at each boundary -- that can be used to encode quantum information.  These edge modes can be qubits or qudits (higher-dimensional generalisations of qubits), depending on the dimension of the representation $V_g$.  As the two boundary modes must remain well-separated, for the purposes of quantum computation it is useful to restrict to only one of them; we choose the left edge mode.  If the chain is sufficiently long, the right mode may be ignored.  For convenience of description, however, we shall consider finite chains and provide a conceptually simple way to terminate the right boundary such that it can be subsequently ignored from the description.

To isolate the left edge mode for our encoding, we terminate the right edge by a symmetric coupling $h_{\text{edge}}$ to an additional new particle.  This coupling will remove the right edge from the description, allowing us to focus purely on the left edge mode.  The new particle is required to transform under a projective representation of the symmetry $V_g^{\text{phys}}$ from the same cohomology class as the left edge mode; see Fig.~\ref{fig:edge}.  This coupling will be spatially local but may have to act on a number of sites up to the injectivity radius \cite{perez2006matrix} of a matrix product state representation of the ground state such that it is possible to achieve the desired coupling on the edge mode by acting only on the physical level. The key importance of this boundary fixing is that  the global symmetry action changes as
\begin{equation}
U_g^{\otimes N} \rightarrow U_g^{\otimes N}\otimes V_g^{\text{phys}}.
\end{equation}
Hence the symmetry now acts on the entire chain by a projective representation with the same cohomology class as the remaining left edge mode. By making the global symmetry projective, it allows us to directly identify the restriction of the global symmetry action on the ground space with a \emph{single} irreducible projective representation $V_g$.  Isolating one edge also has the effect of forcing the ground space degeneracy to be exact as the ground space now forms an irreducible projective representation of the symmetry group.  In contrast, when there are two edge modes there may be an energy splitting that shrinks exponentially with the system size due to a weak coupling between the edges that allows the two irreducible projective representations to couple into a direct sum of unitary representations with slightly different energies. Along with convenience of description, avoiding this coupling has the added advantage that it prevents phase errors accumulating due to the energy splitting of the quasi-degenerate levels of the ground space.	

\subsection{Symmetry-protection of the encoding}

A defining feature of 1D SP phases \cite{else2012symmetry} is the identification of the global symmetry's action within the ground space and the action of a projective symmetry on an emergent edge mode, as
\begin{equation}
\label{eq:iden}
U_g^{\otimes N}\otimes V_g^{\text{phys}} \restr_{\text{ground space}} \sim V_g \,.
\end{equation}
This identification persists throughout the whole SP phase and we find it natural to think of these projective symmetries as `logical' operators acting on the information encoded within the ground space.

Any deformation of the Hamiltonian that maintains the symmetry and a uniform lower bound on the spectral gap gives rise to an adiabatic evolution that must remain in the same SP phase.  Hence, the identification given in Eq.~\eqref{eq:iden} remains valid throughout such an evolution, and because of this, we can talk about the same projective symmetry action $V_g$ on the edge mode throughout the evolution. As we maintain the symmetry at every point of the adiabatic deformation, the resulting evolution within the ground space must commute with this symmetry $V_g$.  If the edge symmetry $V_g$ is irreducible, then by Schur's lemma any adiabatic deformation within the SP phase must necessarily act as the identity operation on the encoded information (up to a global phase) and hence we say that the information encoded into the ground space is protected by the symmetry of the SP phase.  (We do not consider the case where $V_g$ is reducible, although we note that the methods developed in the context of decoherence free subsystems \cite{zanardi1,zanardi2,lidar1998decoherence,lidar2012review} can be used to generalise our results to the reducible case.)

In summary, in this section we have considered adiabatic deformations that strictly preserve the symmetry of the Hamiltonian. However, if an evolution changes the way the symmetry acts, then the arguments given above do not necessarily hold. It is precisely this fact that allows us to perform nontrivial unitary gates upon encoded states using adiabatic evolutions, which will be described in the coming sections.

\section{Elementary gates}
\label{sec:ElementaryGate}

In this section, we demonstrate how an adiabatic evolution involving terms at the boundary between a non-trivial SP phase and a trivial one can result in the implementation of a protected quantum logic gate acting on the fractionalised edge mode.  This gate provides a generalisation of the holonomic gates described in Ref.~\cite{renes2011holonomic}.

We first illustrate the functioning of this gate using a model Hamiltonian that is representative of the SP phase, and subsequently show that the gate properties are generic throughout the phase and independent of the microscopic details of the Hamiltonian.  For simplicity, we will only consider models of SP spin chains with nearest neighbor interactions. (This condition will hold for all 1D models after sufficient real space renormalization.)  We use the notation $H_{i,i+1}$ to indicate a Hamiltonian term acting on sites $i$ and $i+1$, such that the full Hamiltonian $H = \sum_i H_{i,i+1}$ describes a spin chain in a nontrivial SP phase.  We explicitly single out the interaction term $h_{\text{edge}}$ that couples the spins near the right edge (site $N{-}1$) to a fractional particle (that transforms under the projective representation $V_g^{\text{phys}}$) effectively fixing the relevant boundary condition. (This could be thought of as modelling the left edge mode of a semi-infinite SP chain.)
We require that the Hamiltonian terms commute with the symmetry 
\begin{align}
\left[ H_{i,i+1}, U_g^{\otimes N}\otimes V_g^{\text{phys}}\right] &= 0,\quad \forall\ i,g\ ,\\ 
\left[ h_{\text{edge}}, U_g^{\otimes N}\otimes V_g^{\text{phys}}\right] &= 0,\quad \forall\ g \ ,
\end{align}
and that the degenerate ground space transforms under an irreducible projective representation of the symmetry
\begin{equation}
U_g^{\otimes N}\otimes V_g^{\text{phys}} \restr_{\text{ground space}} \sim V_g\,, 
\end{equation}
corresponding to the free, left edge mode.

We now construct a spin chain possessing a phase boundary, with a trivial symmetric phase on the left of the boundary and a non-trivial SP phase on the right.  We model the trivial phase by introducing a uniform symmetric field $F$ that acts on a single site, satisfies $\left[F,U_g\right]=0,\ \forall g$, and possesses a nondegenerate ground state $\Ket{\chi}$.  Due to the symmetry condition, the ground state $\Ket{\chi}$ carries a one-dimensional representation, i.e., a character $\chi: G \rightarrow U(1)$, such that $U_g\Ket{\chi}=\chi(g)\Ket{\chi}$.  Writing $F_i$ for $F$ at site $i$ in tensor product with identity elsewhere, the uniform field Hamiltonian $H_F:=\sum_i F_i$ has a unique, symmetric ground state and hence lies in a trivial symmetric phase.   

\begin{figure}[t]
\center
\includegraphics[width=0.95\linewidth]{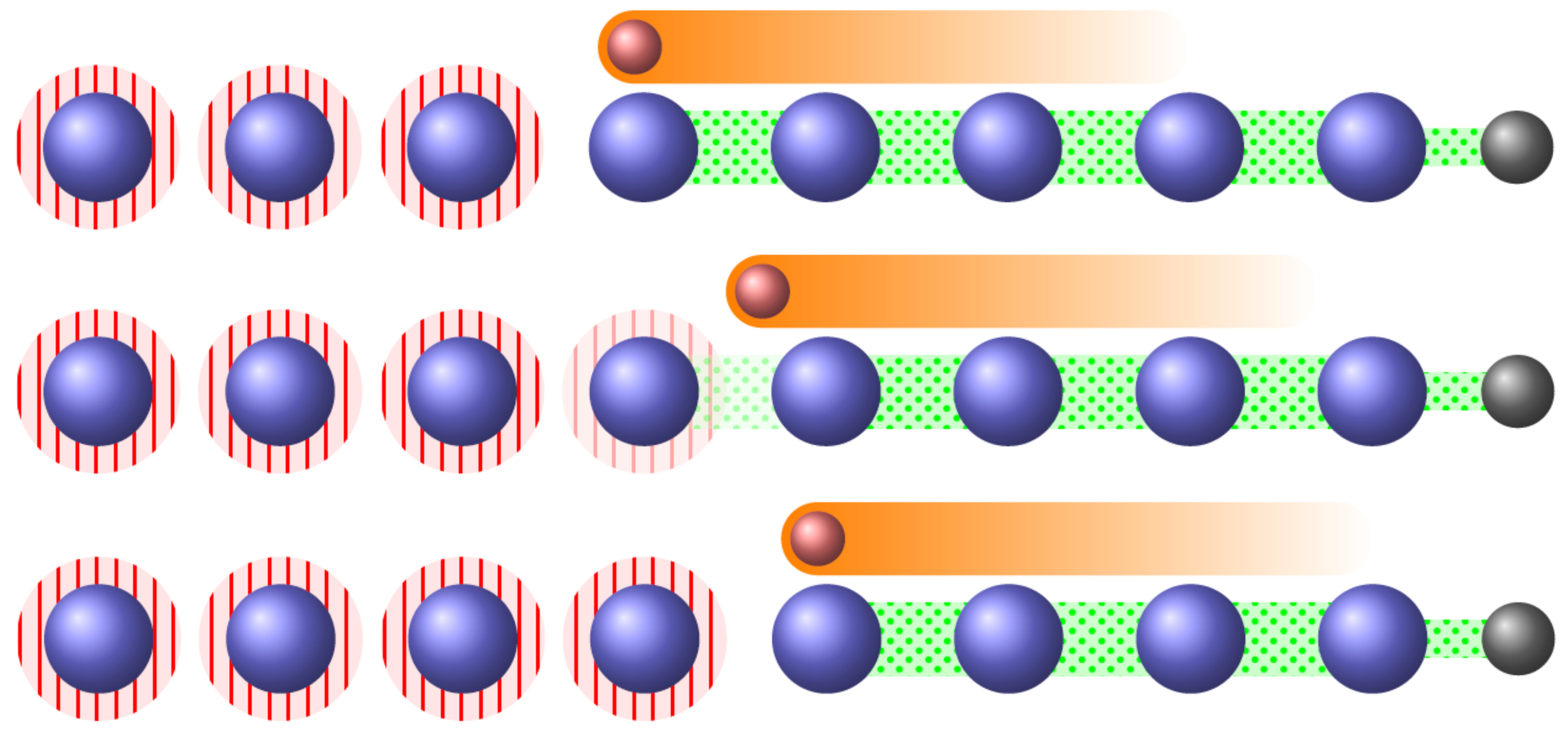}
\caption{Snapshots of the evolution described by Eq.~\eqref{eq:ua2} at times $H_{8,3}(0)$, $H_{8,3}(T/2)$ and $H_{8,3}(T)$. The notation used is the same as that in Fig.~\ref{fig:edge}, with local fields depicted by (red) lines.}
\label{fig:chainholonomy}
\end{figure}

The Hamiltonian describing a boundary, localized at some site $j$, between a trivial symmetric phase to the left of $j$ and a SP phase to the right (with far right boundary fixed) is then given by
\begin{equation}
H_{N,j}:= \sum_{i=0}^{j-1} F_i + \sum_{i=j}^{N-2} H_{i,i+1} +h_{\text{edge}}\ .
\label{eq:ua}
\end{equation}
States in its ground space take the form of a tensor product between a trivial symmetric phase described by $|\chi\rangle^{\otimes j}$ on sites $i<j$ and a SP ordered ground state on sites $i\geq j$.
The symmetry acts within the ground space as
\begin{align}
U_g^{\otimes N}\otimes V_g^{\text{phys}} \restr_{\text{ground space}} &= \chi(g)^j\  U_g^{\otimes N-j}\otimes V_g^{\text{phys}} \restr_{\text{ground space}} \nonumber \\
&\sim \chi(g)^j V_g\,.
\label{eq:sym11}
\end{align}
As discussed in the previous section, we assume that the projective representation $V_g$ is irreducible.  

The elementary gate is performed by adiabatically moving the phase boundary one site to the right along the chain, from $j$ to $j+1$, as depicted in Fig.~\ref{fig:chainholonomy}. This is achieved by turning off the two body Hamiltonian interaction $H_{j,j+1}$ and turning on a symmetric field $F_j$.  The adiabatic evolution is governed by the time-dependent Hamiltonian
\begin{equation}
H_{N,j}(t)= \sum_{i=0}^{j-1} F_i + f(t) F_j + g(t) H_{j,j+1} + \sum_{i=j+1}^{N-2} H_{i,i+1} +h_{\text{edge}}
\label{eq:ua2}
\end{equation}
with $f(0)=g(T)=0$ and $g(0)=f(T)=1$. We require $T=\Omega\left(1/\Delta^3\right)$ to ensure adiabaticity, where $\Delta$ is the minimum spectral gap of $H_N(t)$ as $t$ is varied.  (To be precise, we also require standard smoothness conditions on $f$ and $g$~\cite{jordan2008quantum}.) 

The adiabatic evolution induced by the parametrised Hamiltonian of Eq.~\eqref{eq:ua2} is designed to commute with the symmetry and to preserve the ground space.  We use this property to enact a logical transformation on the SP edge mode by decoupling the spin at site $j$ into a symmetric state $\Ket{\chi}$ in tensor product with the remaining nontrivial SP phase on sites $j{+}1,\dots,N{-}1$. This evolution moves the phase boundary and the edge mode one site to the right along the chain while simultaneously multiplying the projective symmetry action on the edge mode by a phase $\chi(g)^{-1}$.  That is, at the start of the evolution the ground space carries an irreducible projective representation of the symmetry group given by $V_g$ as in Eq.~\eqref{eq:iden}, and at the end of the evolution carries the irreducible representation $V'_g = \chi(g)^{-1} V_g$.  This rephasing does not change the cohomology class of the projective representation and hence does not alter the class of SP ordered phase to the right of the phase boundary.  However, it does allow the evolution within the ground space $W$ to act nontrivially on the encoded information as specified by its intertwining of the two irreducible projective representations
\begin{equation}
W V_g W^{\dagger} = \chi(g) V_g,\quad \forall\ g \,,
\end{equation}
or, equivalently,
\begin{equation}
W V_g = \chi(g) V_g W, \quad \forall\ g.
\label{eq:wv}
\end{equation}
In the case where $\chi$ is the trivial character $\chi=1$, Eq.~\eqref{eq:wv} together with Schur's lemma imply that $W=c\ \mathbb{I},\ \exists c\in U(1)$  (since we are considering an irreducible representation). In contrast, nontrivial characters may generate nontrivial evolutions within the encoded subspace, but in that case we cannot simply invoke Schur's lemma to calculate the evolution.  

For the remainder of this section, we will characterise the evolutions $W$ when $\chi\neq 1$.  We emphasise that our results characterising the evolutions depend only on the symmetry properties of the Hamiltonians and their ground states (which persist throughout a SP phase), and not on any specific description of these Hamiltonians as in Eq.~(\ref{eq:ua2}).  We first recast the conditions of Eq.~\eqref{eq:wv} in terms of the fixed point of a particular channel. We then present some basic properties of this channel and show that the fixed points of different channels arising from the same SP phase form a faithful projective representation of the abelianisation of the symmetry group (presented as Theorem~\ref{thm:abelianisation}). Finally we show how to construct fixed points from tensors satisfying a certain natural symmetry condition and give a simple example.

\subsection{The group of elementary gates}

We begin by expressing the conditions of Eq.~\eqref{eq:wv} in terms of the fixed point of a channel defined by the one-dimensional representation (character) $\chi\in \hat{G}$, where $\hat{G}$ is the group of one dimensional representations, and then show that the fixed point is unique up to a phase.

\begin{prp}
The conditions $W V_g = \chi(g) V_g W$ for all $g$ are equivalent to the condition that the matrix $W$ is a fixed point of the channel
\begin{equation}
\label{eq:chan}
\Gamma_{\chi} (\, \cdot\, ) := \frac{1}{|G|} \sum_{g\in G}\chi(g) V_g (\,\cdot\, ) V_g^{\dagger} \,, \end{equation}
\end{prp}
\begin{proof}
To show the equivalence in the forward direction, we note that if a matrix $W$ satisfies $W V_g = \chi(g) V_g W$ for all $g$ then
\begin{align}
\Gamma_{\chi} (W) &= \frac{1}{|G|}\sum_g \chi(g) V_g W V_g^{\dagger} \nonumber \\
&= \frac{1}{|G|}\sum_g \chi(g) \chi^{-1}(g) W V_g V_g^{\dagger} \nonumber \\
&= \frac{1}{|G|}\sum_g W  \nonumber \\
&= W\,.
\end{align}
Conversely, if $\Gamma_{\chi} (W) =W$, then $\forall\ g \in G$
\begin{align}
V_g W &= V_g \frac{1}{|G|} \sum_h \chi(h) V_h W V_h^{\dagger} \nonumber \\
&= \frac{1}{|G|}\sum_h \chi(h) V_{gh} W V_h^{\dagger} \nonumber \\
&= \frac{1}{|G|}\sum_{h'} \chi(g^{-1}h') V_{h'} W V_{g^{-1}h'}^{\dagger} \nonumber \\
&= \chi(g^{-1}) \frac{1}{|G|}\sum_{h'} \chi(h') V_{h'} W V_{h'}^{\dagger}V_{g^{-1}}^\dagger \nonumber \\
&= \chi^{-1}(g) \Gamma_{\chi} (W) V_{g} \nonumber \\
&= \chi^{-1}(g) W V_{g} \,.
\end{align}
\end{proof}

\begin{lem}
The fixed point of the channel $\Gamma_\chi$ is unique up to a phase.
\label{lem}
\end{lem}
\begin{proof}
We consider any two fixed point solutions $W,\ W'$ for the same channel $\Gamma_\chi$ and combine them to form the matrix $W(W')^{\dagger}$. It is easy to see that $W(W')^{\dagger}$ must be a fixed point of the channel $\Gamma_1$ (the channel given by the trivial character $\chi=1$) as
\begin{align}
\Gamma_1\left( W(W')^{\dagger} \right) &= \frac{1}{|G|}\sum_g V_g W(W')^{\dagger} V_g^{\dagger} \nonumber \\
&= \frac{1}{|G|}\sum_g \chi^{-1}(g) W V_g \left( V_g W' \right)^{\dagger} \nonumber \\
&= \frac{1}{|G|}\sum_g \chi^{-1}(g) W V_g \left( \chi^{-1}(g) W' V_g \right)^{\dagger} \nonumber \\
&= \frac{1}{|G|}\sum_g |\chi^{-1}(g)|^2 W V_g V_g^{\dagger} (W')^{\dagger} \nonumber \\ 
&=  W(W')^{\dagger}.
\end{align}
Hence we must have $W(W')^{\dagger}=c\ \mathbb{I},\ \exists c\in U(1)$ by Schur's lemma and the irreducibility of the projective representation $V_g$.
\end{proof}

These results show that the maps $W$ satisfying $\Gamma_\chi(W) = W$ are determined, up to a phase, by the characters $\chi$ of $G$.  We will therefore label the fixed point of $\Gamma_\chi$ by the matrix $W_\chi$, where we make an arbitrary choice of multiplicative phase factor.

The next theorem reveals the group structure of the maps $W_\chi$.  

\begin{thm}\label{thm:abelianisation}
The fixed points $\{W_\chi, \chi\in \hat{G}\}$ of the channels $\Gamma_\chi$ form a faithful projective representation of the abelianisation of the symmetry group $G$.  
\label{thm2}
\end{thm}
\begin{proof}  Let $W_\chi$, $W_\varphi$ be the fixed points corresponding to characters $\chi$, $\varphi$ respectively.  Note the one dimensional representations of $G$ form an abelian group under pointwise multiplication, i.e., $[\chi \cdot \varphi] (g) := \chi(g) \varphi(g)$.  Then, by Lemma~\ref{lem}, the fixed point $W_{\chi\cdot \varphi}$ of the channel $\Gamma_{\chi\cdot \varphi}$ is also unique up to a phase. Now observe 
\begin{align}\label{eq:wawbsol}
\Gamma_{\chi \cdot \varphi} (W_{\chi}W_{\varphi}) &= \frac{1}{|G|}\sum_g \chi(g) \varphi (g) V_g W_{\chi}W_{\varphi} V_g^{\dagger} \\ 
&= \frac{1}{|G|}\sum_g \chi(g) \varphi (g) \chi^{-1}(g) W_{\chi}V_g W_{\varphi} V_g^{\dagger} \nonumber\\ 
&= \frac{1}{|G|}\sum_g \varphi (g) \varphi^{-1}(g) W_{\chi} W_{\varphi} V_g V_g^{\dagger} \nonumber\\ 
&= W_{\chi} W_{\varphi}.\nonumber
\end{align}
and so $W_{\chi}W_{\varphi}$ is also a fixed point of the channel $\Gamma_{\chi\cdot \varphi}$.
Hence $W_{\chi}W_{\varphi}=\alpha(\chi,\varphi) W_{\chi \cdot \varphi}$ for some phase function $\alpha:G'\times G' \rightarrow U(1)$, where $G':=G/\left[G,G\right]$ is the abelianisation (maximal abelian quotient) of $G$.  Because $W_\chi$ was itself only defined up to a phase, the possible solutions form a projective representation of $G'$.  Furthermore, this representation is faithful, since the identity is a fixed point of $\Gamma_{\chi}$ if and only if $\chi=1$ by the orthonormality of distinct characters.
\end{proof}

The abelianisation appears because it is isomorphic to the group of one dimensional representations of $G$.  This abelianisation is a natural object in this context, since $\Gamma_{\chi\cdot \varphi}=\Gamma_{\varphi \cdot \chi}$.  We emphasise, however, that being a projective representation of an abelian group, the unitaries $W_\chi$ need not commute, only `commute up to a phase', i.e., $W_\chi W_\varphi = e^{i\theta} W_\varphi W_\chi$.  

Theorem~\ref{thm:abelianisation} represents the central result of this section, in that it determines the set (actually a group) of unitary logic gates that can be performed by adiabatically shifting the boundary between a SP ordered phase and a trivial symmetric phase, as described by the Hamiltonian in Eq.~(\ref{eq:ua2}).  This directly generalises the approach of Ref.~\cite{renes2011holonomic} to arbitrary groups, and moreover can be thought of as a generalisation of this scheme to generate open loop holonomic gates via the manipulation of the phase boundary at the edge of a SP ordered spin chain.

\subsection{Constructing elementary gates from tensors}

In the following, we provide an explicit construction of the elementary gates using tensor network language, and in doing so connect the elementary gates to the so called by-product operators \cite{else2012symmetry,gross2007novel} that arise in measurement-based quantum computation (MBQC) using a SP phase~\cite{else2012symmetry2}.  We begin by demonstrating some properties of the channels $\Gamma_\chi$.

\begin{lem}
The channels $\Gamma_\chi$ are orthogonal projectors. Equivalently, $\Gamma_{\chi} \circ \Gamma_{\varphi} =  \delta_{\varphi,\chi} \Gamma_{\chi}$.
\label{lem1}
\end{lem}
\begin{proof} 
For an arbitrary matrix $M$, we have
\begin{align}\label{eq:chcomp}
&\Gamma_{\chi} \circ \Gamma_{\varphi}\ (M) = \Gamma_{\varphi} \bigl( \Gamma_{\chi} (M) \bigr) \nonumber \\ 
&= \frac{1}{|G|^2}\sum_g \varphi(g) V_g \Bigl( \sum_h \chi(h) V_h M V_h^\dagger \Bigr) V_g^\dagger \nonumber\\
&= \frac{1}{|G|^2}\sum_{g,h} \varphi(g)\chi(h) V_{gh} M V_{gh}^\dagger \nonumber\\
&= \frac{1}{|G|^2}\sum_{a,b} \varphi(b)\chi(b^{-1}a) V_a M V_a^\dagger \nonumber\\
&= \frac{1}{|G|^2}\sum_{a,b} \varphi(b)\chi(b^{-1}) \chi(a)V_a M V_a^\dagger \nonumber\\
&= \frac{1}{|G|^2} \Bigl[ \sum_{b} \varphi(b)\chi^{*}(b) \Bigr] \sum_{a} \chi(a)V_a M V_a^\dagger \nonumber\\
&= \delta_{\varphi, \chi} \Gamma_\chi (M) \,,
\end{align}
where we have made use of the orthonormality of distinct characters. 
\end{proof}

Finding the fixed points $\Gamma_{\chi} \left( W_\chi \right) = W_\chi$ of the channel is essentially the same problem as finding the symmetric subspace of the representation $\chi(g) V_g \otimes V_g^{*}$, where $*$ denotes complex conjugation. We can see that the maps $\Pi_\chi := \frac{1}{|G|} \sum_g \chi(g) V_g \otimes V_g^{*}$ form a set of orthogonal projections by Lemma~\ref{lem1}. The representation $\chi(g) V_g \otimes V_g^{*}$ acts upon the two virtual degrees of freedom associated to any single site in a symmetric matrix product state (MPS) representation of a renormalisation fixed point ground state in a SP phase; see Fig.~\ref{fig:ax1} and Refs.~\cite{chen2011classification,schuch2011classifying,else2012symmetry}.

\begin{figure}[ht]
\center
\includegraphics[width=0.78\linewidth]{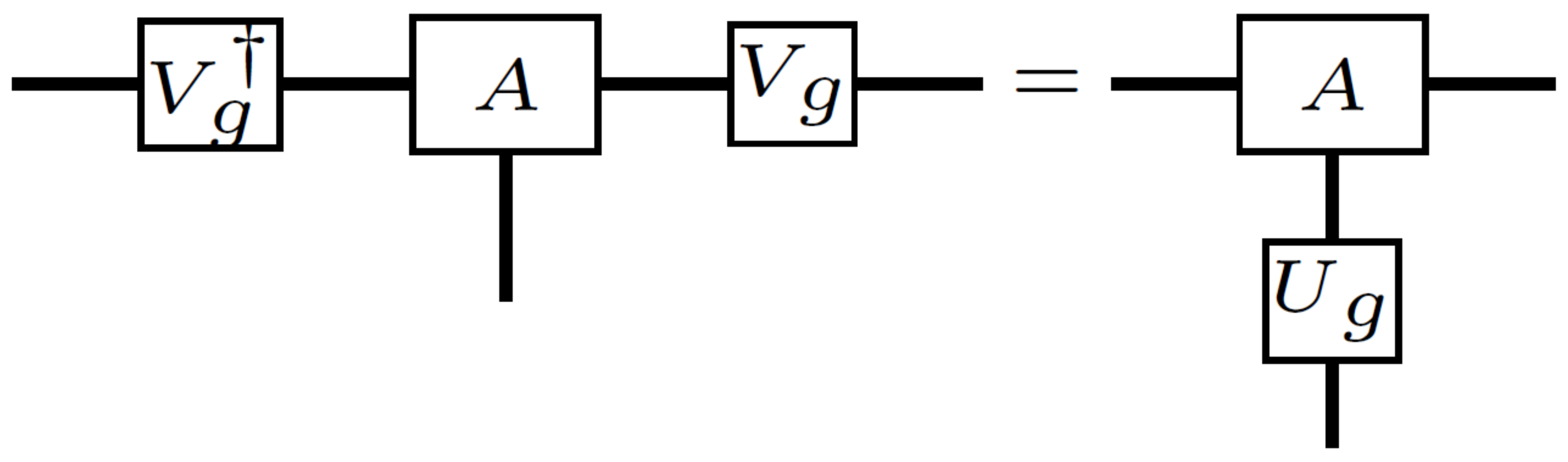}
\caption{The symmetry condition satisfied by a renormalisation fixed point MPS tensor of a SP phase.  A tensor $A$ satisfying this condition is referred to as a \textit{$(U_g,V_g)$-symmetric  MPS tensor}.}
\label{fig:ax1}
\end{figure}

We can equivalently understand the virtual entangled states acted upon by this representation in terms of a Clebsch-Gordon matrix coupling the two projective representations to a single representation on the physical level. We note that the Clebsch-Gordon matrix is essentially the same object as the fixed point MPS tensor but without any projection onto a subspace at the physical level. 

We define a \textit{$\left(U_g,V_g\right)$-symmetric MPS tensor} $A_{\alpha,\beta}^i$, $i=1,\ldots,d$ and $\alpha,\beta = 1,\ldots,D$ to be a tensor that obeys the symmetry condition depicted in Fig.~\ref{fig:ax1}, i.e.,
\begin{equation}
\sum_{\gamma,\delta=1}^D \left(V_g\right)_{\alpha,\gamma} A_{\gamma,\delta}^i \left(V_g^\dagger\right)_{\delta,\beta} = \sum_{j=1}^d A_{\alpha,\beta}^j \otimes \left( U_g\right)_{j,i}
\end{equation}
for all $g \in G$, where $U_g$ is a unitary representation and $V_g$ a projective representation of $G$.
This definition encompasses the cases of fixed point MPS tensors and Clebsch-Gordon matrices. 

The matrix $A\left[ \psi \right]_{\alpha,\beta}$ can be constructed by projecting the physical leg of the tensor $A_{\alpha,\beta}^i$ onto the state $\Ket{\psi} \in \mathbb{C}^d$ as depicted in Fig.~\ref{fig:ax}, i.e.,
\begin{equation}
A\left[ \psi \right]_{\alpha,\beta} : = \sum_{i=1}^d A_{\alpha,\beta}^i \Braket{i | \psi}.
\end{equation}

\begin{thm}
Given a $\left(U_g,V_g\right)$ symmetric MPS tensor $A$ (for the same $U_g,V_g$ appearing in Eq.~\eqref{eq:sym11}) the fixed point of $\Gamma_\chi$ is $W_\chi = c A[\chi]$ for some arbitrary phase $c\in U(1)$.

\label{thm1}
\end{thm}
\begin{proof}
By the symmetry condition in Fig.~\ref{fig:ax1} and the transformation of $\Ket{\chi}$ under $U_g$ we have the property depicted in Fig.~\ref{fig:ax} from which it is clear that $A\left[\chi \right]$ is a fixed point of $\Gamma_\chi$. Hence by Lemma~\ref{lem} we have $W_\chi=c~A\left[\chi\right]$ for some $c\in U(1)$.
\end{proof}
\begin{figure}[ht]
\center
\includegraphics[width=0.9\linewidth]{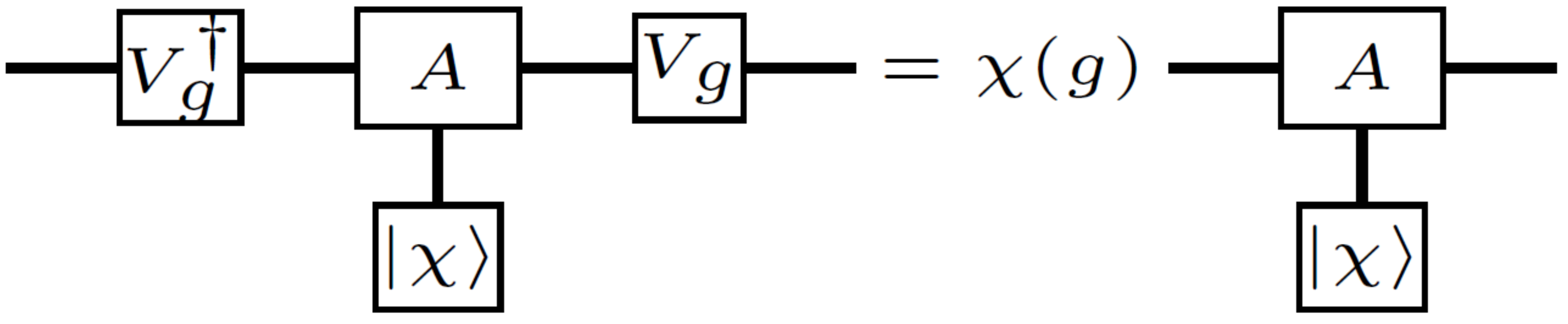}
\caption{A fixed point solution $W_\chi=A\left[\chi\right]$ of the channel $\Gamma_\chi$, constructed from a $(U_g,V_g)$-symmetric MPS tensor $A$, and describing the adiabatic evolution within the ground space.}
\label{fig:ax}
\end{figure}

From this theorem, we also see, conversely,  how to construct a $(U_g,V_g)$-symmetric MPS tensor $A$ from a set of solutions $W_\chi$ to the channels $\Gamma_\chi$ for a given set of one dimensional representations $\{\chi\}$ of $G$, i.e.,
\begin{equation}
A=\sum_{\chi} W_\chi \otimes \bra{\chi}
\label{eq:syma}
\end{equation} 
as depicted in Fig.~\ref{fig:ax2}. For nonabelian groups, however, this limits us to only considering unfaithful, abelian representations $U_g$ on the physical level. This result is essentially the same as that given in Refs.~\cite{else2012symmetry,else2012symmetry2} except that we have not required the additional condition that the projective representation $V_g$ is maximally noncommutative.

\begin{figure}[ht]
\center
\includegraphics[width=0.62\linewidth]{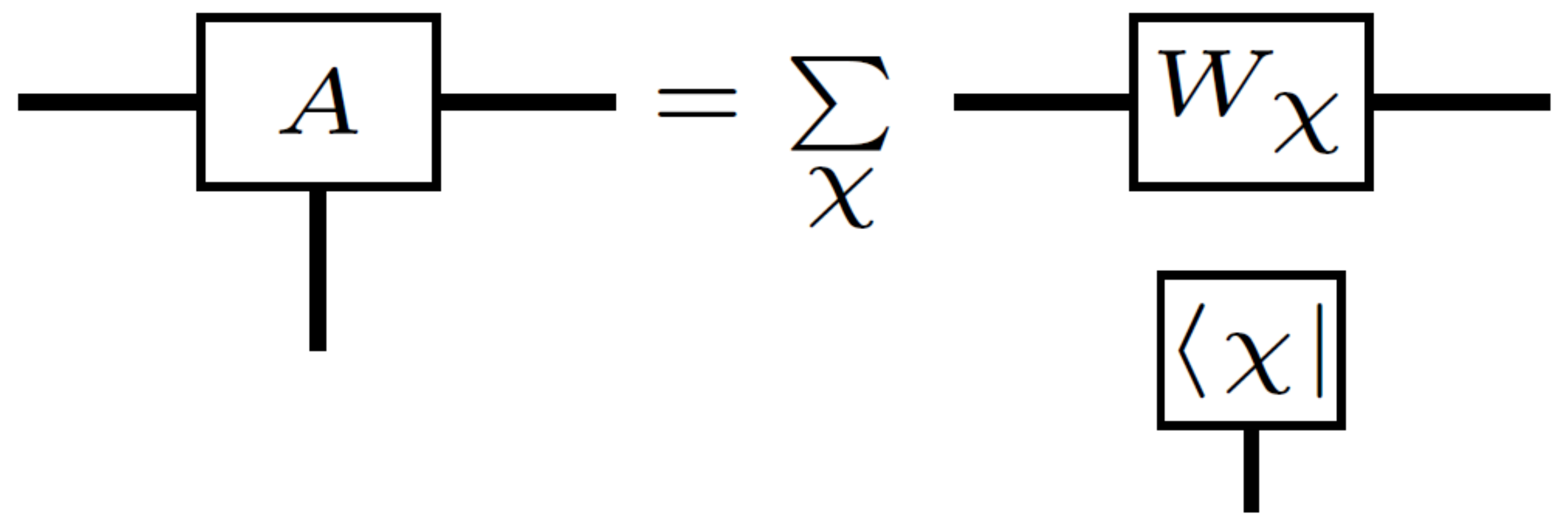}
\caption{A $(U_g,V_g)$-symmetric MPS tensor $A$ constructed from a fixed point solution $W_\chi$ of the channel $\Gamma_\chi$.}
\label{fig:ax2}
\end{figure}

\subsubsection*{Example:  Haldane phase}

We now present two simple examples of the gate construction described above.  These examples are fixed-point states in Haldane phases of spin chains with different local physical dimension.

Firstly the Affleck-Kennedy-Lieb-Tasaki (AKLT) fixed-point state~\cite{AKLT}, which is representative of the spin-1 Haldane phase~\cite{haldane1983continuum} protected by the group of $\pi$-rotations about two orthogonal spatial axes (isomorphic to $\mathbb{Z}_2\times\mathbb{Z}_2$) represented by spin-1 $\pi$-rotation matrices. The choice of the axes of rotation are arbitrary so let us fix them to be $\hat{x}$ and $\hat{z}$. A nontrivial irreducible projective representation of this group is given by the Pauli matrices $\left\{ I,\sigma^x,\sigma^y,\sigma^z \right\}$. The AKLT tensor $A$ is constructed by taking the sum over $m=x,y,z$ of a tensor product of the Pauli matrix $\sigma^{m}$ with the 0-eigenstate $\ket{m}$ of the corresponding spin-1 operator, i.e., $S^m\Ket{m}=0$, 
\begin{equation}
A=\sum_{m\in\{x,y,z\}} \sigma^m \otimes \bra{m}\ .
\end{equation}
The states $\Ket{m}$ transform under the spin-1 $\pi$-rotations via the character $\chi_m$ defined by $\chi_m(1)=\chi_m(m)=1$, with all other values $-1$.  This tensor
$A$ is easily seen to be a $\left(U_g,V_g \right)$ symmetric MPS tensor (where $U_g$ is the group generated by the spin-1 $\pi$-rotations about $\hat{x}$ and $\hat{z}$, and $V_g$ is generated by the Pauli matrices) because it is of the general form in Eq.~\eqref{eq:syma} and Fig.~\ref{fig:ax2}. Hence each matrix $A\left[ m \right]=W_m$ is the fixed point of the corresponding channel $\Gamma_m$.  We note that these maps $W_m$ form the Pauli group of a single qubit, which is a faithful projective representation of the symmetry group $\mathbb{Z}_2\times\mathbb{Z}_2$ that protects the phase.

Our second example is given by the cluster state, a fixed-point state in a Haldane phase of a chain of sites each containing two spin-$\frac{1}{2}$ particles, which is protected by the group generated by Pauli $\sigma^x$ matrices applied simultaneously to all odd or even particles respectively (which is isomorphic to $\mathbb{Z}_2\times\mathbb{Z}_2$). We group pairs of spins together, and the cluster state can be written as a MPS with local tensor 
\begin{equation}
A=\openone\otimes \bra{++} + \sigma^x \otimes \bra{+-} + \sigma^z \otimes \bra{-+} - i \sigma^y \otimes \bra{--} \,.
\end{equation}
The states $\ket{\pm \pm}$ clearly transform as characters of $\mathbb{Z}_2\times\mathbb{Z}_2$ under the on-site representation $U_g$ generated by $\sigma^x\otimes \openone$, $\openone \otimes \sigma^x$. One can readily verify the MPS tensor $A$ is $\left(U_g,V_g \right)$ symmetric (for $U_g$ the on-site representation mentioned above and $V_g$ the projective representation generated by the Pauli matrices) and is of the particular form given in Eq.~\eqref{eq:syma} and Fig.~\ref{fig:ax2}. This immediately yields that the maps $W_\chi$ are Pauli matrices.

\subsection{Section summary}

In this section, we have introduced the basic adiabatic evolution that moves the boundary between a trivial symmetric and SP phase along a chain by one site. This evolution has the effect of moving the information encoded into the edge mode at the boundary spatially while simultaneously applying a nontrivial evolution determined by the symmetric field applied.  We have shown that the abelian group action (one-dimensional representation) $\chi$ of the symmetry group on the trivial ground state completely specifies this evolution on the encoded quantum information.  The one-dimensional representations form an Abelian group, and the associated evolutions form a projective representation of this abelianisation of the symmetry group $G$. We also showed how to construct explicit solutions starting from the symmetric tensors of an exact MPS description of a fixed-point ground state.  Again, we emphasise that our construction makes use only of symmetry group properties, and thus properties of the zero-temperature phase, and not of any specific form of the Hamiltonian.

\section{Extending the set of logic gates}
\label{sec:GeneralGates}
As we have shown, symmetry-respecting adiabatic evolution of the fractionalised edge mode at the boundary of a SP phase allow us to perform certain quantum logic gates, specifically, those given by a projective representation of the abelianisation of the symmetry group.  This is related to similar results for MBQC using the ground state of a non-trivial SP phase \cite{else2012symmetry,else2012symmetry2}, wherein the identity gate can be performed perfectly throughout the phase, up to some unitary correction operators that depend on the measurement results.  These special quantum logic gates `commute up to a phase', and in addition they act only on a single encoded qudit.  As such, they cannot form a universal gate set for quantum computation.  We now explore ways to supplement this elementary gate set with additional operations to make it universal.  

To implement additional logic gates in the standard MBQC model, and similarly in the AQT model, measurements or fields that do not respect the symmetry of the phase are employed.  One would not expect these unprotected operations to function uniformly well throughout a phase, but instead would depend on the microscopic details of the Hamiltonian.  For this reason, we do not explore this direction further.  We point the interested reader to two recent proposals that employ additional ingredients to endow such gates with protection, by imposing additional symmetries (that are not both unitary and on-site)~\cite{prakash2014universal} or via the inclusion of an additional decoupling procedure between gates~\cite{miller2014resource}. 

An alternate approach to achieve more general logic gates is to employ several different types of SP matter, each with an inherent symmetry protecting a distinct gate that together make up a universal set.   The techniques of Ref.~\cite{renes2011holonomic} can be used to illustrate this idea using a particular, well understood example based upon the spin-1 Haldane phase, described in detail in Appendix~\ref{app:haldane}. For a universal set of single qubit gates, we exploit the fact that both $\mathbb{Z}_2\times\mathbb{Z}_2$ $\pi$-rotation symmetry about two orthogonal spatial axes and full $SO(3)$ rotation symmetry protect the same edge modes on a chain of spin-1 particles.  By choosing three distinct SP phases, each protected by a different embedding $\mathbb{Z}_2\times\mathbb{Z}_2\subset SO(3)$, we can implement the Hadamard gate (requiring one distinct phase) and the $\pi/8$-rotation gate (requiring a combination of two others) that together yield a universal single qubit gate set.  We note that the protected gates performed using a single embedding require us to explicitly break the $SO(3)$ symmetry down to a subgroup, and thus are not compatible with any other embedding.  To perform multiple gates protected by different embeddings requires different SP phases, to achieve this one could make use of large (bulk) regions of the requisite nontrivial SP phase for each gate, connected to one another by bulk regions which satisfy the full $SO(3)$ symmetry.  This perspective naturally leads us to the quantum transistors discussed in the next section.

To achieve a universal gate set for many qubits, we complement the single qubit gates with a nontrivial entangling gate (equivalent to a controlled-phase gate up to local Pauli operations) on two qubits achieved within a SP phase of two coupled spin-1 chains. This SP phase is shown to be protected by a semi-direct product symmetry $(\mathbb{Z}_2\times\mathbb{Z}_2)\rtimes\mathbb{Z}_4\subset SO(3)\times SO(3)$ embedded within the rotation group of two decoupled chains, which protects the same edge modes as the full $SO(3)\times SO(3)$ symmetry in this case. This example is based upon the work in Ref.~\cite{renes2011holonomic} but goes beyond, as explained in detail in Appendix~\ref{app:haldane}, with proof of the previously unknown fact that the full gate set is symmetry-protected. Our approach differs from Ref.~\cite{renes2011holonomic} in that they consider a continuous, dynamical embedding $\mathbb{Z}_2\times\mathbb{Z}_2\subset_t SO(3)$; we have avoided this and use three fixed embeddings.

Although the above example is specific to the group $SO(3)$, it is shown in Ref.~\cite{else2013hidden} that for both $SO(2k+1)$ and $SU(k)$ there exist discrete, abelian subgroups protecting the same SP phase as the full continuous symmetry.  Hence we expect in these cases that a similar approach could be employed to generate a desired set of symmetry-protected gates by using different embeddings of the discrete subgroup that protects the same edge modes as the full continuous group.

The single qubit untwisted cluster gate described in Ref.~\cite{bacon2012adiabatic} also falls into our framework (see \cite{else2012symmetry} for a description of the analogous MBQC description of the cluster state), while the twisted cluster gates of Ref.~\cite{bacon2012adiabatic} do not as they employ fields that do not respect the $Z_2 \times Z_2$ symmetry.  Hence we expect such gates to function well only near exact fixed points of a SP phase such as the AKLT point or cluster state. 
Furthermore, their method for implementing two qubit gates does not possess an irreducible edge mode. This raises the possibility of finding two qubit couplings that do not support irreducible edge modes but are still capable of generating unique logical evolutions outside the natural gate set of the specific SP model.  However, there may not be the same robustness of the edge mode encoding if the coupling Hamiltonian in such a process is varied, since the edge modes are no longer irreducible and hence not protected throughout the phase.

\section{Realising a Transistor}
\label{sec:Transistor}

In the previous sections, we showed that quantum logic gates acting on the information encoded in the fractionalised edge modes defined at the boundary of a nontrivial SP phase can be performed using local adiabatic evolutions.  In this section, we will generalise the notion of an adiabatic quantum transistor, defined by Bacon, Crosswhite and Flammia~\cite{bacon2012adiabatic}, to show that our logic gates can also be performed by applying a global field across the whole chain simultaneously in the adiabatic limit. 

For a spin chain of length $N$, let $H^{\rm SP}_V$ be a Hamiltonian in a non-trivial SP phase whose ground state transforms via the projective representation $V_g$ of $G$ associated to a fractionalized left boundary, as in Sec.~\ref{sec:SPencoding}.  Let $H^{\rm triv}_\chi$ be a Hamiltonian describing a uniform field applied to the chain on all but the right boundary, with a non-degenerate ground state that transforms under $G$ with character $\chi$.  As a model, let 
\begin{align}
  H^{\rm SP}_V &= \sum_{i=0}^{N-2} H_{i,i+1}+h_{\text{edge}}\,, \\
  H^{\rm triv}_\chi &= \sum_{i=0}^{N-1} F_i \,,
\end{align}
as described in Sec.~\ref{sec:ElementaryGate}, although our results are not restricted to Hamiltonians of this form and apply throughout the respective SP and trivial phases of any two such Hamiltonians.  We consider an adiabatic evolution initiated entirely within the SP phase and to which a global field is then applied, given by 
\begin{equation}
\label{eq:spqat}
  H_N(t)=f(t) H^{\rm triv}_\chi + g(t) H^{\rm SP}_V\,,
\end{equation}
where again $f(0)=g(T)=0$ and $g(0)=f(T)=1$.  To ensure adiabaticity, we again require $T=\Omega\left(1/\Delta^3\right)$, with $\Delta$ the minimum spectral gap of $H_N(t)$.  With this time-dependent Hamiltonian, it is clear that we are driving the system through a symmetric phase transition from a SP phase to a trivial symmetric phase. Furthermore, the unitary evolution on the ground space will be the same $W_\chi$ no matter what point of the SP phase we start in, and hence the evolution is truly a property of the whole phase.  %
\begin{figure}[t]
\center
\includegraphics[width=0.95\linewidth]{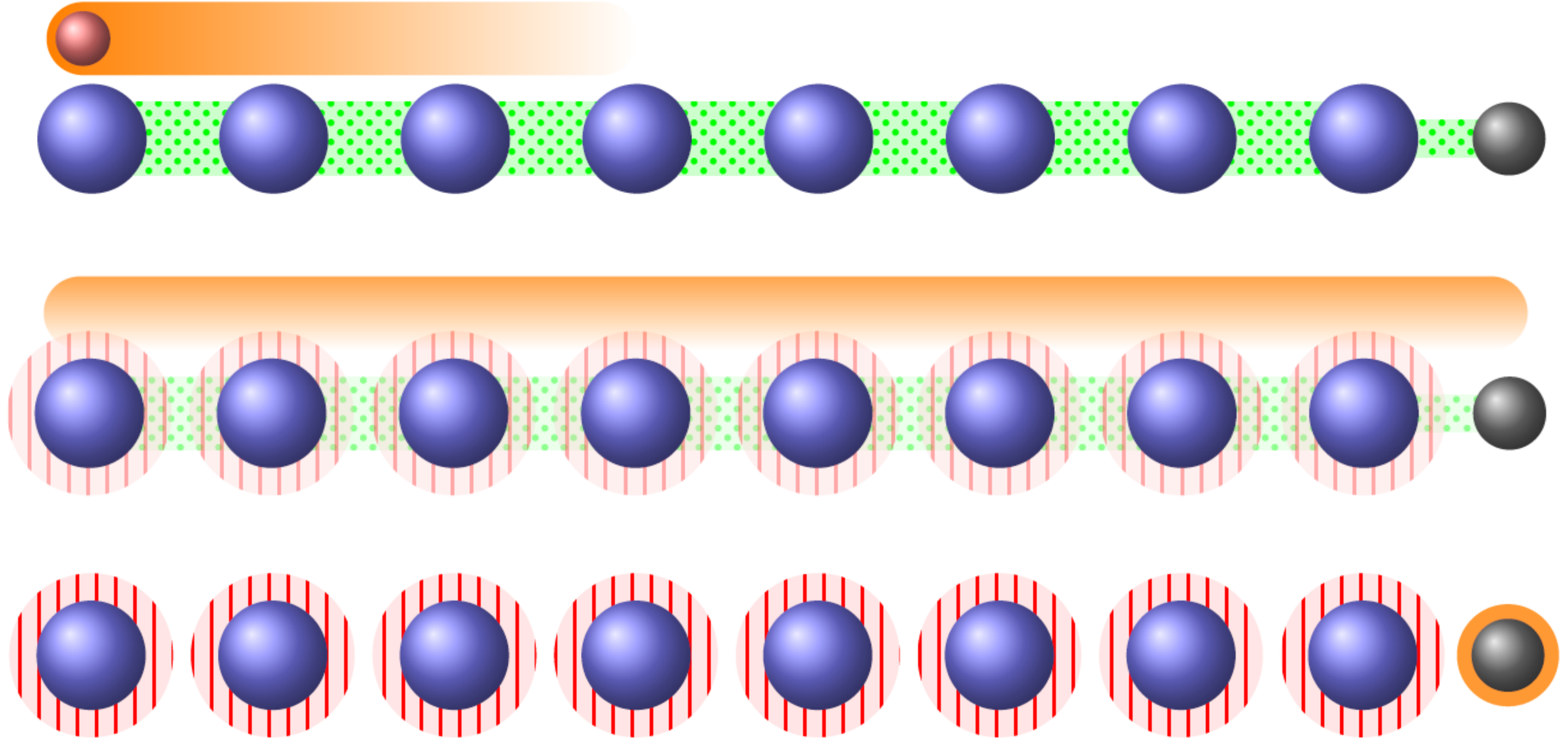}
\caption{Snapshots of the evolution described by Eq.~\eqref{eq:spqat} at times $H_{8}(0)$, $H_{8}(T/2)$ and $H_{8}(T)$. The notation used is the same as that in Fig.~\ref{fig:edge} and Fig.~\ref{fig:chainholonomy}. The orange shading indicates the support of the quantum information that is initially localised in the left edge mode, which then becomes delocalised across the whole chain and ends up in the physical particle at the right edge.}
\label{fig:chaintransistor}
\end{figure}

We note that the assumption of adiabaticity is significant, given that, in the thermodynamic limit ($N{\to}\infty$), this evolution will pass through a phase transition.  The total time $T$ will be much longer than that of the single spin evolution of Eq.~\ref{eq:ua2} as the minimum gap $\Delta$ will approach zero as the system size grows (although remaining nonzero for any finite size). The exact rate at which the gap closes will determine the efficiency with which we are able to simulate circuits on single or multiple qudits (dependent upon the coupling Hamiltonian) and the required time must not increase by more than a polynomial factor for the scheme to be viable. Hence, we require that the gap can be bounded from below by the inverse of a polynomial in the system size, $\Delta = \Omega\left(1/\poly(N)\right)$. 

There are relatively few techniques for bounding the spectral gap of arbitrary Hamiltonians, and so one does not expect to prove efficiency in general for all SP phases.  However, as noted by Bacon \textit{et al.}~\cite{bacon2012adiabatic}, a proof for a universal gate set would be an important step.  
For 1D systems we take a different approach to argue for efficiency, making use of the fact that all ground states of gapped 1D Hamiltonians are well approximated by MPS~\cite{hastings2007area,landau2013polynomial}, and that the circuit to construct an MPS is in general linear in the chain length $N$.  Even under completely symmetric evolution, any SP state (with an exact MPS representation) can be mapped to a product state by a circuit of depth $O(N)$~\cite{huang2014quantum} by taking the standard circuit to construct the MPS representation of the state and applying its inverse. This circuit will also commute with the symmetry, provided that the MPS tensors are symmetric. Furthermore, this scaling should be optimal as we expect one edge must communicate with the other edge to complete the disentangling map to a product state.  In principle, then, there could exist $O(N)$ length symmetric adiabatic evolutions that map from a trivial phase to a SP phase (and similarly between SP phases). As such, we expect that generic adiabatic evolutions generated by the application of a field for a period scaling polynomially with the system size $O(\poly (N))$ will pass from SP to trivial phase, and these may even be as fast as $O(N)$.

Another interesting and relevant aspect of the efficiency issue is addressed in the work of Dziarmaga and Rams~\cite{dziarmaga2010adiabatic,dziarmaga2010dynamics}, investigating the adiabatic traversal of symmetry-breaking phase transitions in one-dimensional quantum spin models.  In comparison to the uniform application of a global field driving the system through a phase transition, it is found that sweeping a spatially-varying field profile sufficiently slowly across a chain leads to a polynomial speed up in the time required for adiabatic traversal. In addition, the slow sweeping can lead to an exponential suppression in the density of thermal errors. We expect that similar results should apply at least in the case of an abelian symmetry group with maximally noncommutative factor system~\cite{else2013hidden} as there exist locality-preserving mappings between the symmetry broken and SP phases in this case.  This approach may also become natural when trying to implement a full circuit, as we envision sweeping a uniform field sufficiently slowly across a network of SPAQTs arranged to simulate the circuit so as to only address a number of SPAQT gates in parallel at each time step, rather than applying the field simultaneously to the whole complex network (by sweeping we hope to avoid the kind of trade off described in \cite{antonio2013trade}).  Finally, we note that the adiabatic sweep can be run in reverse, which could be advantageous if, as is suggested in Ref.~\cite{dziarmaga2010adiabatic,dziarmaga2010dynamics}, the thermal errors generated by the sweeping field propagate ahead of the phase transition wave front. In this situation, the excitations will be swept away from the edge mode on the SP portion of the chain and into the trivial phase.

\section{Errors and their effects}
\label{sec:Errors}

The SPAQT offers some natural robustness to a variety of errors that could occur in a quantum computation.  In this section, we will survey the various error channels for this scheme, following the outline of the discussion in Bacon \textit{et al.}~\cite{bacon2012adiabatic}. See also \cite{cesare2014adiabatic} for a more thorough analysis of error channels relevant to holonomic adiabatic quantum computation using a many-body ground state.

As our model is technically adiabatic open loop holonomic quantum computation, the result that general holonomic quantum computation can be performed fault-tolerantly~\cite{oreshkov2009fault} are relevant, although the construction in Ref.~\cite{oreshkov2009fault} does not lend itself directly to our framework.  Nonetheless, it shows in principle that a holonomic scheme for universal quantum computation such as ours can be made fault tolerant.  The question of whether we can design such a fault-tolerant construction while preserving the desirable physical properties of our scheme is open. 

As a first requirement, our scheme relies on the existence of robust non-trivial SP phases, for example the Haldane phase~\cite{Kestner2011,janani2014} and topological insulators~\cite{moore2010}, which have been observed in both condensed matter and cold atomic systems.  In particular, we require that SP phases should still exist for small non-symmetric Hamiltonian perturbations and small nonzero temperatures. We note that recent studies have shown that localization through disorder in SP phases can provide some robustness to the fractionalized edge modes~\cite{bahri2013localization}.

Each individual SPAQT is inherently thermodynamically protected from all symmetric errors due to the irreducibility of the edge mode representation within the ground space and the energy gap to the excited states.  In this sense, our encoding is essentially a decoherence free subspace~\cite{lidar1998decoherence,lidar2012review} for symmetric errors. This property can be combined with dynamical decoupling pulses~\cite{viola1998dynamical} implementing all global symmetry transformations in order to symmetrize the noise operators to a certain order in perturbation theory.  Performing such sequences would then provide thermodynamic protection from these now symmetrized errors, as they must act as the identity on the ground space to the same order of approximation. However, this scheme may not suit the adiabatic implementation in practice as the dynamic decoupling requires active application of fast pulse sequences to implement global symmetries throughout the evolution.

Any errors that have the sole effect of changing the energy eigenspace of the chain during the adiabatic evolution should be equivalent to having an excited state (or superposition of excited states) at the end of the computation, where the Hamiltonian consists of purely uniform local fields.
We restrict our attention to the case of a single excited state, as any superposition can be collapsed by measuring each spin in the basis of the field being applied to it. We can understand such an error as causing some of the spins to end up in excited eigenstates of this applied field. Provided that the eigenspaces of the field are nondegenerate, the excited state should transform as some irreducible representation $\varphi: G\rightarrow U(1)$ of the symmetry. If we label the excited state by $\Ket{\varphi}$ then the effect of this error is precisely to implement the gate $W_\varphi$ in place of the $W_\chi$ that would have occurred without the error (there may also be a global phase factor due to the $W_\varphi$ matrices forming a projective representation). Upon measuring all the spins in the trivial phase at the end of the evolution, we can in principle determine what excitation errors have occurred during the computation and furthermore collapse a superposition of such errors into the energy eigenbasis of the applied field.

A subtlety we have overlooked thus far is the possibility of a small accidental, adiabatic deformation away from the desired final state of the applied field.  While this would seem easy to suppress in practice by simply applying a stronger uniform local field it does not cause any change in the intended logic gate so long as the deformation is symmetric. The only effect this may have is to entangle the encoded information at the edge of the SP chain with some spins in the trivial phase near the phase boundary. This could necessitate some operation to disentangle the logical information, again, measuring the spins in the basis of the applied field should suffice with high probability.

\subsection{Protection through delocalisation}

For the remainder of the section, we speculate about some possible fault tolerant properties of the encoding we use at different points in a SP phase. The ground state encoding is associated to a gapless edge mode and it is known to be localised to the edge in the following sense: there is a renormalisation fixed point of the phase in which the information is strictly localised to a single physical site and as we follow a symmetric, adiabatic path this mode spreads out up to the point where it persists across the majority of the chain at a phase transition. This implies that the encoding will possess different degrees of inherent robustness to local errors. As noted by Bacon \textit{et al.}~\cite{bacon2012adiabatic}, at an exact fixed point and at the decoupled end point of the adiabatic evolution in Eq.~\eqref{eq:spqat}, the encoded information is essentially as unprotected as a bare qubit. They propose a solution to this by scheduling the adiabatic evolution to spend a minimal amount of time at the beginning and end of the computation, where the gap is almost constant and the information is unprotected. They go a step further and conjecture that the encoded information is inherently robust to local errors during the middle of the adiabatic evolution where it is maximally delocalised over the bulk of the chain, see Fig.~\ref{fig:chaintransistor}.

We note that even for points in a SP phase that are a constant distance from the fixed point, the encoded information could spread over a constant number of sites that is sufficiently large to protect against errors that act independently on single physical spins. An analysis of the general case is complicated by the fact that, if the parent Hamiltonian consists of commuting terms, then the information lives precisely on the single physical edge spin and is therefore unprotected. Hence any inherently robust encoding must have non-commuting Hamiltonian terms, and analysing the precise properties of such an encoding would be difficult. We further propose that during computation, when no measurements are necessary, one should take advantage of the inherent robustness of points in the phase where the edge mode is spread out. In particular we consider starting and finishing the computation at such points, rather than the points with exactly localised encoded information described by Eq.~\eqref{eq:spqat}. 

\subsection{Nonsymmetric errors}

Finally, we consider the most general nonsymmetric error operators.  For a generic Hamiltonian in a SP phase, there will be a nontrivial dispersion relation, hence, localized errors will propagate across the chain.  We believe that it should still be possible to deal with these errors by globally cooling the system while sweeping a field to implement the computation. In this case the errors can only propagate a certain mean path length determined by the temperature of the memoryless cooling reservoir to which the chain is coupled. Then with high probability the region which could possibly be effected by each error is of constant size in time and space and should be uncorrelated with other errors. Hence we expect that such an error should be correctable by simulating standard fault tolerant circuit constructions such as Ref.~\cite{aharonov2008fault} with the SPAQTs. 

It may be possible to formalize the above analysis by treating the spin chain as weakly coupled to a bath where the open system dynamics can be described by a master equation. In some such cases it has been shown that the light cone of information spread can asymptote to a finite region \cite{descamps2013asymptotically}. Furthermore the presence of some weak disorder in the system could have a similar effect in localizing the excitations caused by errors such that they can be corrected using the procedure described above~\cite{bahri2013localization}.

\section{Conclusion}
\label{sec:Conclusions}
We have argued that material properties of SP phases make them natural systems to use when designing adiabatic quantum transistors, in loose analogy to the use of semiconductor materials in building classical transistors. 
We have proposed an understanding of the operation of an adiabatic transistor in terms of driving a spin chain through a phase transition from a symmetry-protected phase to a trivial symmetric phase. This perspective also extends the understanding of Hamiltonians that lead to adiabatic quantum transistor gates from finely tuned exact models to whole SP phases of matter, thereby further reducing the control requirements of the scheme.
We would particularly like to highlight the fact that the logical transformations implemented by a SPAQT depend only upon symmetry properties that are universal to a whole SP phase.

We further hope that our general approach can be applied to a broad range of situations to characterise useful properties of particular fine tuned parent Hamiltonians in terms of more robust and universal properties of whole quantum phases.

Finally, we put forward the conjecture that our scheme may be adaptable to exploit the inherent protection of 2D topologically ordered surface states of 3D SP or topological bulk materials \cite{burnell2013exactly}, thus achieving inherent fault tolerance of the information encoded into the edge mode.
We also conjecture that our schemes extends, in a natural way, to currently engineerable topological wires with Majorana edge modes~\cite{kitaev2001unpaired,sau2010,mourik2012signatures,churchill2013superconductor} which are fermionic analogs to the (bosonic) SP spin chains studied here. 

\section*{Acknowledgements}
We thank Steven Flammia, Andrew Doherty, Joseph Renes, and Gavin Brennen for helpful comments. DW thanks Beno{\^\i}t Descamps for the suggestion to look at the symmetry conditions as a channel. This research was supported by the ARC via the Centre of Excellence in Engineered Quantum Systems (EQuS), project number CE110001013, and by the U.S. Army Research Office. 


\appendix

\section{SPAQT with the Haldane phase}
\label{app:haldane}

In this appendix, we analyze the model of Ref.~\cite{renes2011holonomic} as an example of a SPAQT.  Much of the material presented here is a review of the details of Ref.~\cite{renes2011holonomic}, expressed in the language and notation of the present paper for clarity; however, several new results are included.  First, we demonstrate that the two qubit gate of Ref.~\cite{renes2011holonomic} is associated with a non-trivial symmetry protected phase of two spin chains (while this was previously only known to hold for the single qubit gates).  Second, we demonstrate how combining evolutions under several different discrete symmetry groups give rise to a universal gate set.  Together, these facts confirm the claim that the model of Ref.~\cite{renes2011holonomic} falls into our framework. 

\subsection{Encoding in the Haldane phase} \label{appendixd} \label{sec:heischain}

In this section, we review the spin-1 Heisenberg chain proposed in Ref.~\cite{renes2011holonomic} for holonomic quantum  computation. We consider a 1D chain of spin-1 particles that locally interact via a pairwise, symmetric, antiferromagnetic coupling, favoring local anti-alignment of neighboring spins. A two-body, nearest neighbor Hamiltonian which describes such an interaction with full $SO(3)$ rotation symmetry is
\begin{equation}
\label{so3ham}
H^{\text{Haldane}} = J  \sum_i \left( \vec{S}_i\cdot\vec{S}_{i+1} - \beta \left( \vec{S}_i \cdot \vec{S}_{i+1} \right)^2 \right) \,,
\end{equation}
where $J>0$ for an antiferromagnetic coupling.  This Hamiltonian describes the Heisenberg model at $\beta=0$ and the AKLT model at the point $\beta=-\frac{1}{3}$.  These two models lie within a common phase -- the Haldane phase -- corresponding to the range $-1<\beta<1$ in the $SO(3)$ symmetric Hamiltonian \eqref{so3ham}, and characterised by fractionalized spin-$\half$ boundary degrees of freedom throughout this parameter range. This is an instance of a SP phase protected by the $SO(3)$ rotation symmetry, with the boundary spins described by the spin-$\half$ projective representation of $SO(3)$.

The states of the spin-$\half$ boundary degrees of freedom label a fourfold degeneracy (in the thermodynamic limit) in the ground states of the spin chain. To be precise we note that for any finite chain there is a small splitting between the energy eigenvalues of the set of ground states corresponding to the singlet and triplet states of the edge modes. This splitting decays exponentially as the size of the system grows, while the gap to the first excitation converges to a non-zero value. This is due to the general property that correlations decay exponentially in gapped ground states, causing the strength of the interaction between the two edge modes to decay accordingly. 
\begin{figure}[t]
\center
\includegraphics[width=0.95\linewidth]{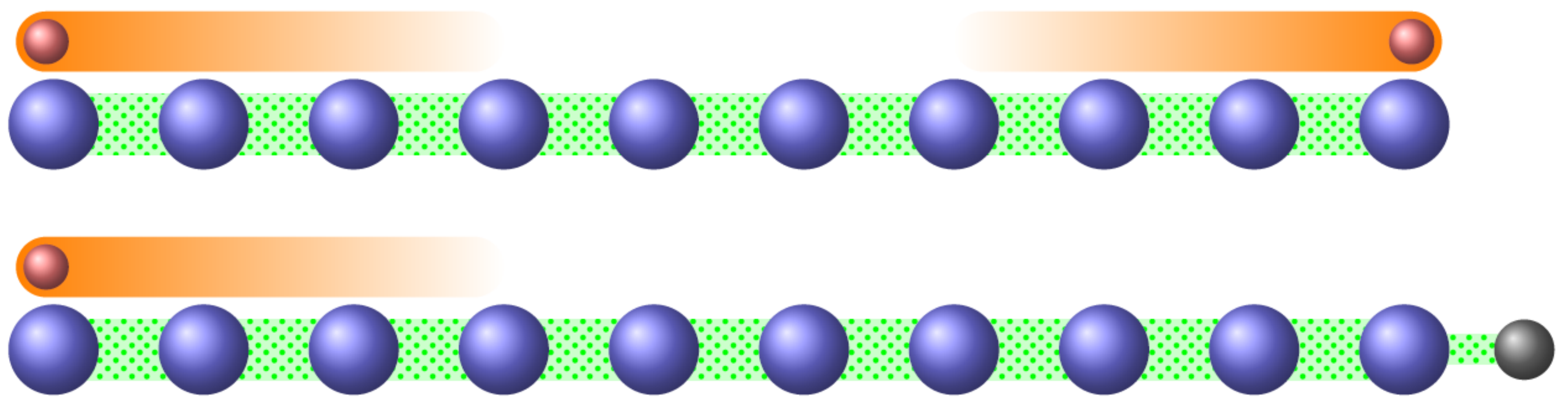}
\caption{A Haldane chain of spin-1 particles with open (top) and and one fixed (bottom) boundary conditions.}
\end{figure}
As described in Sec.~\ref{sec:SPencoding}, we consider coupling one boundary to a real spin-$\half$ which possesses a nontrivial, projective representation of the symmetry group with a cohomology class label, inverse to that of the boundary mode. This effectively purifies that edge mode and removes the fractional degree of freedom. The product of the emergent mode Hilbert space with that of the real spin-$\half$ is now equivalent to a linear representation of the symmetry. Hence, coupling at one boundary breaks the (near) fourfold degeneracy of the ground states down to a twofold degeneracy (which is exact, even for a finite sized system) corresponding to a single spin-$\half$ boundary mode.
\begin{equation}
\label{so3term}
H_{m,n}^{\text{Haldane}} = J  \sum_{i=m}^{n-1} \left( \vec{S}_i\cdot\vec{S}_{i+1} - \beta \left( \vec{S}_i \cdot \vec{S}_{i+1} \right)^2 \right)+J\vec{S}_n\cdot\vec{s}_{n+1}
\end{equation}
The purification of one boundary, effectively fixing that degree of freedom, reduces the dimension of the degenerate ground space to two. We identify the logical Pauli operators on this subspace with global conserved quantities, generated by the symmetries of the Hamiltonian
\begin{align}
\label{eq:consquant}
\Sigma_{n}^{\hat{m}} &=\left( \bigotimes_{j=1}^{n}\exp (i \pi S^{\hat{m}}_j) \right) \otimes  \exp (i \frac{\pi}{2} \sigma ^{\hat{m}} ) \\
&= \bigotimes_{j=1}^{n} \left( I-2 (S_j^{\hat{m}})^2  \right) \otimes  i \sigma ^{\hat{m}} 
\end{align}
These operators form a nontrivial projective representation of $SO(3)$, with the same second cohomology label as the 2D spin-$\half$ representation.  The logical operators of the single encoded qubit are $Z_L=\Sigma_{n}^{\hat{z}}$ and $X_L=\Sigma_{n}^{\hat{x}}$. The encoded spin-$\half$ degree of freedom spanned by the eigenstates of these operators within the degenerate ground state is identified with the state of the gapless boundary mode. This encoding persists throughout the SP phase since it relies only on conserved quantities generated by the symmetries of the whole phase and the ground state degeneracy that is protected by this symmetry.

We note that the Haldane phase can be protected by the Abelian subgroup $D_2=\mathbb{Z}_2\times\mathbb{Z}_2 \subset SO(3)$~\cite{KennedyTasaki}.  We can think of this group $D_2$ as being embedded in the natural $SO(3)$ symmetry, corresponding to a subgroup generated by $\pi$-rotations about two orthogonal, spatial axes.  The relaxation of the symmetry condition allows us to explicitly consider $D_2$ symmetry respecting local fields of the form $(\vec{S}^{\hat{m}})^2$ acting on single spin sites along the three spatial axes which define the embedding of $D_2\subset SO(3)$. For an explicit embedding of $D_2$ generated by $\pi$-rotations about two orthogonal axes $\hat{m}$, $\hat{m}^{\perp}$, we can use the local fields $(\vec{S}^{\hat{m}})^2$, $(\vec{S}^{\hat{m}^{\perp}})^2$ and $(\vec{S}^{\hat{m}\times\hat{m}^{\perp}})^2$ without breaking the symmetry.

In the next section we will use these symmetry respecting fields to generate logical evolutions of the encoded qubits.

\subsection{Single-Qubit Gates}
\label{singlequbitgates}

In this section, we continue reviewing the results of Ref.~\cite{renes2011holonomic}, demonstrating that single-qubit Pauli rotations can be performed by adiabatically decoupling a single spin from the chain while applying a $D_2$-symmetry respecting field to it.  We describe the operation of the gates using the exact Heisenberg Hamiltoniann $\beta=0$ for simplicity, but the arguments are based purely on symmetry arguments and hold equally well if we use any Hamiltonian throughout the Haldane phase.

The qubit encoded in the free edge of the ground state by the $X_L$ and $Z_L$ operators can be manipulated by adiabatically decoupling a single spin from the end of the chain while applying a local field to it. This unitary evolution forces  the decoupled spin into the ground state of the local field operator. For a field aligned along the $\hat{z}$-axis, this evolution is governed by the following time dependent Hamiltonian
\begin{equation}
\label{eq:szhol}
H_n(t)=f(t) J (S_1^{\hat{z}})^2 + g(t) J \vec{S}_1 \cdot \vec{S}_2 + H_{2,n}^{\text{Haldane}}
\end{equation}
where $f$ and $g$ are monotonic functions with: $f(0)=g(T)=0$ and $f(T)=g(0)=1$. Note that the addition of the  $(S_1^{\hat{z}})^2 $ field fixes one axis of the embedding $D_2\subset SO(3)$ to be the $\hat{z}$ axis. To complete a nontrivial closed holonomy with $D_2$ symmetry we apply a local field along another axis, orthogonal to $\hat{z}$. The choice of the second axis specifies the embedding of $D_2 \subset SO(3)$. The particular choice of a field along the $\hat{x}$ axis identifies $D_2\subset SO(3)$ with the group of $\pi$-rotations about the $\hat{x}, \hat{y}, \hat{z}$ axes. The full holonomy is then described by the Hamiltonian
\begin{equation}
\label{1qbitholonomy}
H_n(t)=f_1(t) J (S_1^{\hat{z}})^2 + f_2(t) J (S_1^{\hat{x}})^2 + g(t) J \vec{S}_1 \cdot \vec{S}_2 + H_{2,n}^{\text{Haldane}}
\end{equation}
where $f_1$, $f_2$ and $g$ are smooth functions, piecewise-monotonic on the three time intervals: $[0,T_1]$, $(T_1,T_2]$, $(T_2,T_3]$, with: $f_1(0)=f_1(T_2)=f_1(T_3)=f_2(0)=f_2(T_1)=f_2(T_3)=g(T_1)=g(T_2)=0$ and $f_1(T_1)=f_2(T_2)=g(0)=g(T_3)=1$; see Fig.~\ref{fig:1qbit1}.
This time varying Hamiltonian respects the $D_2$ symmetry throughout the coupling and hence supports the SP phases protected by this symmetry. Consequently the boundary modes persist so long as there is no phase transition in the path of the time dependent Hamiltonian.  (Note that Ref.~\cite{renes2011holonomic} cites strong numerical evidence  that the energy gap remains finite for these evolutions, and hence there is no phase transition.) 
\begin{figure}[t]
\center
\includegraphics[width=0.95\linewidth]{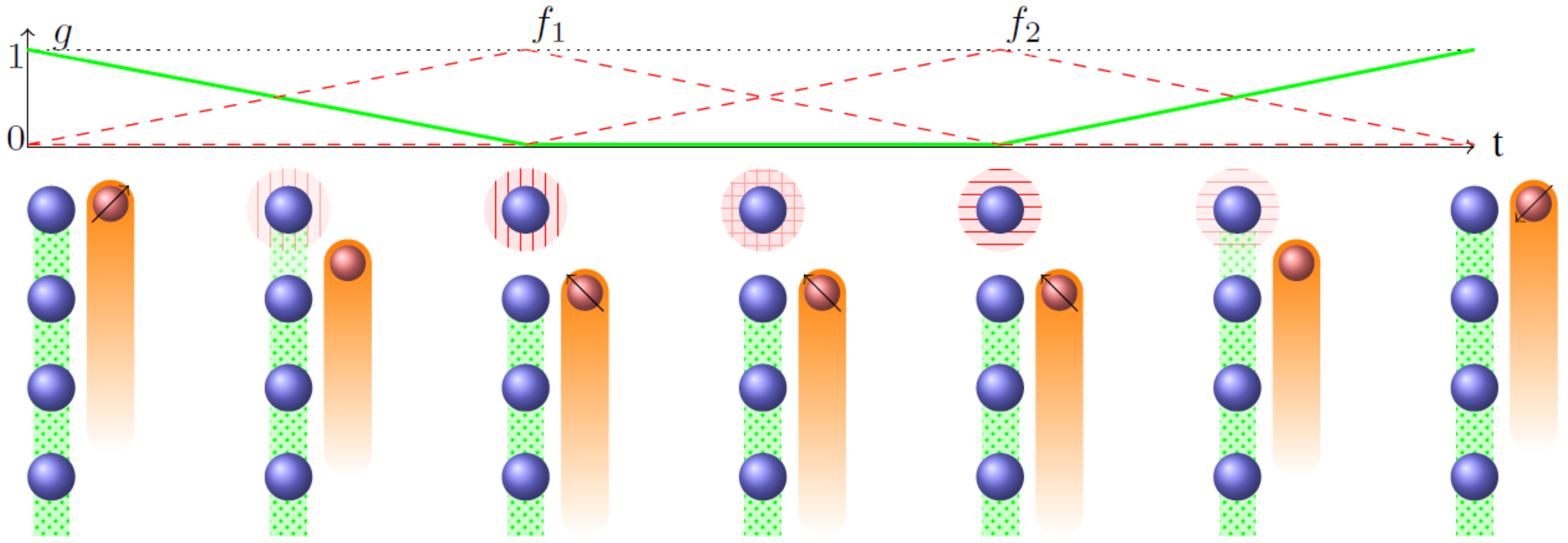}
\caption{The holonomic evolution inducing a single-qubit gate and the coupling strengths throughout the process.}
\label{fig:1qbit1}
\end{figure}
We analyze the action of the holonomy on the encoded spin by making use of the conserved quantities $\Sigma_n^{\hat{x}},\Sigma_n^{\hat{z}}$ generated by the on-site symmetries, which remain constant during the unitary evolution. First we consider the evolution over the interval $[0,T_1]$, as a spin-1 is decoupled from the $n$-chain and the encoded qubit squeezed into a shorter chain of length $(n-1)$. This evolution results in a Pauli $Z_L$ gate on the encoded information.

We fix notation as follows.  A quantum state $|H=0\rangle$ denotes a state in the ground space of $H$, normalised such that this has lowest eigenvalue $0$; for degenerate ground states, additional quantum numbers are used to uniquely specify a state.  For an initial $+1$ eigenstate of $\Sigma_n^{\hat{z}}$, the $\ket{0}_L$ logical state on $n$ spins
\begin{equation}
\ket{\Psi(0)}=\ket{\Sigma_n^{\hat{z}}=1,H_{1,n}^{\text{Haldane}}=0}=\ket{0}_n
\end{equation}
after the adiabatic decoupling becomes
\begin{equation}
\ket{\Psi(T_1)}=\ket{\Sigma_n^{\hat{z}}=1,(S_1^{\hat{z}})^2=0,H_{2,n}^{\text{Haldane}}=0}
\end{equation}
up to a phase factor. This eigenstate represents a spin-1 decoupled from the remaining $(n-1)$ length chain.

To determine the results of this evolution we make use of the conserved quantities on the chain. The ground state of the decoupled spin $\ket{(S_1^{\hat{z}})^2=0}=\ket{S_1^{\hat{z}}=0}$ is a +1 eigenstate of the rotation operator $\exp (i \pi S^{\hat{z}} )$ hence the remaining symmetry operator $\Sigma_{n-1}^{\hat{z}}$ must have eigenvalue +1, so that the value of the total $\Sigma_n^{\hat{z}}$ is conserved. Hence the final state can be written
\begin{align}
\ket{\Psi(T_1)}&=\ket{S^{\hat{z}}=0}\otimes \ket{\Sigma_{n-1}^{\hat{z}}=1,H_{2,n}^{\text{Haldane}}=0}\\
&=\ket{S^{\hat{z}}=0}\otimes\ket{0}_{n-1}
\end{align}
and we see that the encoded $\ket{0}_L$ state is fixed under this evolution. Similarly an initial state $\ket{1}_n$ evolves to ${\ket{S^{\hat{z}}=0}\otimes\ket{1}_{n-1}}$ up to another phase.

Since the evolution fixes the $\ket{0}_L$ and $\ket{1}_L$ states up to possibly different phase factors, it must amount to some rotation about the $\hat{z}$-axis of the Bloch sphere, corresponding to the unitary operator
\begin{equation}
\label{eq:ul}
 U_L=\begin{bmatrix}
  1 & 0  \\
  0 & e^{i\theta} 
 \end{bmatrix}
\end{equation}
up to an irrelevant global phase, and where $\theta \in (-\pi,\pi]$.

To calculate the rotation $\theta$ we consider the evolution of the $X_L$ basis under the decoupling, making use of the fact that $\exp (i \pi S^{\hat{x}} ) \ket{S^{\hat{z}}=0} = - \ket{S^{\hat{z}}=0}$. For the initial +1 eigenstate of $\Sigma_n^{\hat{x}}$, the $\ket{+}_L$  logical state: 
\begin{equation}
\ket{\Psi(0)}=\ket{\Sigma_n^{\hat{x}}=1,H_{1,n}^{\text{Haldane}}=0}=\ket{+}_n
\end{equation}
after decoupling becomes
\begin{align}
\label{eq:+evo}
\ket{\Psi(T_1)}
&= \ket{S^{\hat{z}}=0} \otimes \ket{\Sigma_{n-1}^{\hat{x}}=-1,H_{2,n}^{\text{Haldane}}=0} \nonumber \\
&= \ket{S^{\hat{z}}=0} \otimes \ket{-}_{n-1}
\end{align}
up to a phase factor $e^{i \gamma}$.

We can fully determine the evolution by comparing the two different descriptions in Eq. \eqref{eq:ul} and Eq. \eqref{eq:+evo}. This comparison implies that $U_L \ket{+}_L = \ket{0}_L + e^{i\theta} \ket{1}_L = e^{i \gamma} \ket{-}_L$, which specifies $e^{i\theta}=-1$, hence the rotation is $\theta=\pi$ about the $z$-axis:
\begin{equation}
 U_L=\begin{bmatrix}
  1 & 0  \\
  0 & -1
 \end{bmatrix} \,.
\end{equation}
Hence the full evolution of the logical qubit described by Eq. \eqref{eq:szhol} has been specified to be a $\pi$-rotation about the $\hat{z}$ axis of the Bloch sphere that takes $\ket{+}_L \mapsto \ket{-}_L$ and $\ket{-}_L \mapsto \ket{+}_L$.

An important point is that this whole argument works just as well when we replace $\hat{z}$ and $\hat{x}$ by any pair of orthogonal axes $\hat{m}$ and $\hat{m}_\perp$ and the field $(S_1^{\hat{z}})^2\mapsto(S_1^{\hat{m}})^2$, which would lead to a $\pi$-rotation about the $\hat{m}$ axis along which the local field is aligned.

Since this evolution is unitary, it can equally well be run in reverse, effectively recoupling a spin, initially in the ground state of a local field, to the chain. This increases the length of the chain and reverses the logical evolution of the decoupling process. Hence the recoupling process also causes a $\pi$-rotation about the axis $\hat{m}$ along which the local field is aligned.

Equipped with a description of the decoupling and recoupling processes, we can determine the full evolution described by Eq.~\eqref{1qbitholonomy}. We see that this corresponds to first a $\pi$-rotation about the $\hat{z}$ axis as a spin is decoupled over the period $[T_0,T_1]$, followed by the adiabatic realignment of the local boundary field from the $\hat{z}$ axis to the $\hat{x}$ axis during $[T_1,T_2]$ and finally another $\pi$-rotation about the $\hat{x}$ axis as the spin is recoupled from  $[T_2,T_3]$. Hence the total evolution associated to the holonomy is just a $\pi$-rotation about the $\hat{y}=\hat{z}\times \hat{x}$ axis of the logical Bloch sphere.  

\begin{widetext}

\subsection{A symmetry-protected two-qubit gate}
\label{sec:2qbit}
\label{sec:evo2}

In this section, we review the entangling gate of Ref.~\cite{renes2011holonomic} between the qubits encoded in two separate chains.  This gate uses a similar procedure to the single-qubit evolution, but this time by coupling a pair of physical spins, one from each chain, as they are simultaneously decoupled from their respective chains.  
We then present several new results.  We first calculate the symmetry group (which we call $G_2$) of the two-chain interaction Hamiltonian in Sec.~\ref{sec:symgp}, and investigate the representations (including projective representations) of this group in Sec.~\ref{sec:repsofG2}.  With this symmetry group, we then prove the main result of the appendix, Theorem~\ref{thm:1} in Sec.~\ref{sec:pfnontriv}, that this symmetry group protects a SP phase and hence also protects the two-qubit gate.

\subsubsection{A two-qubit gate}

In this section, we review the operation of the two-qubit entangling gate of Ref.~\cite{renes2011holonomic}.

To simulate more complicated quantum circuits involving multiple qubits we need to be able to generate entanglement between encoded qubits. We do this in a similar way to the single-qubit gates, but this time by brining together two spin chains ($A$ and $B$) and applying two-body interaction terms to a pair of spins at the edge of the chains. We use the particular choice of coupling Hamiltonian $W^{AB}$ introduced in Ref.~\cite{renes2011holonomic} which yields a controlled-$Z$ gate (abbreviated as the `CZ gate') followed by local Pauli operators (by-products) on each individual chain as we decouple the pair of end spins. 
\begin{equation}
\label{2qbit}
H_n(t)=f(t) J\ W^{AB} + g(t) J \left( \vec{S}^{A}_1 \cdot \vec{S}^{A}_2 + \vec{S}^{B}_1 \cdot \vec{S}^{B}_2 \right) + H_{2,n}^{A,\text{Haldane}}+H_{2,n}^{B,\text{Haldane}}
\end{equation}
where $f(0)=g(T)=0$, $f(T)=g(0)=1$ and the symmetric coupling $W^{AB}$ is given by
\begin{equation}
W^{AB} = \left[ (S_1^{\hat{x}})^2-(S_1^{\hat{y}})^2 \right] _A \otimes \left[ S_1^{\hat{z}}\right] _B + \left[ S_1^{\hat{z}} \right] _A  \otimes \left[ (S_1^{\hat{x}})^2-(S_1^{\hat{y}})^2 \right] _B
\end{equation}
\begin{figure}[t]
\center
\includegraphics[width=0.5\linewidth]{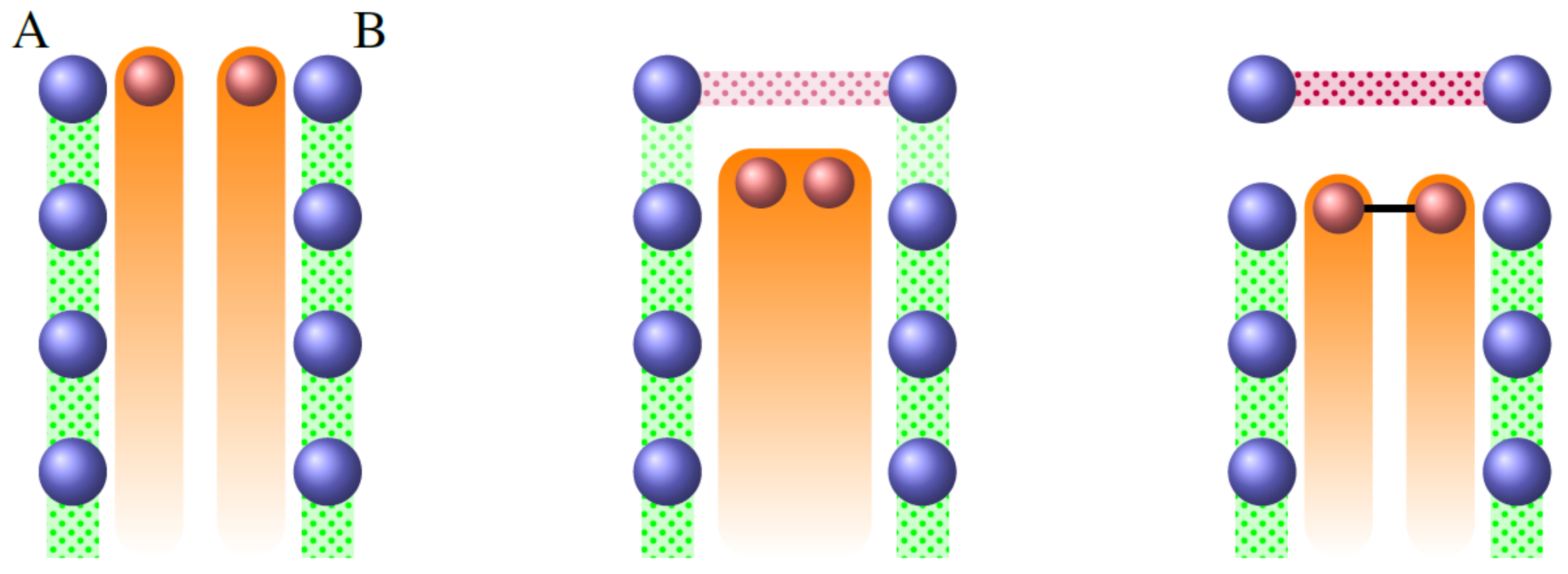}
\caption{The holonomic evolution which induces a an entangling gate on two encoded qubits.}
\label{fig:2qbit1}
\end{figure}

To calculate the evolution of the encoded qubits under the Hamiltonian~\eqref{2qbit} we make use of similar symmetry arguments to those for the single-qubit gate. For this purpose, the symmetry operators
of the interaction term $W^{AB}$ and the conserved operators generated by them upon the full two chains are instrumental. The state of the decoupled end spins determines the evolution of the remaining chains via the conserved quantities. The $W^{AB}$ coupling has the unique groundstate 
\begin{align}
\ket{\xi}=&\frac{1}{2}({-}\ket{1}\ket{1}+\ket{1}\ket{{-}1}+\ket{{-}1}\ket{1}+\ket{{-}1}\ket{{-}1}) 
\end{align}
note that this groundstate is invariant under the full symmetry group of the two chain interaction and hence does not induce any symmetry breaking in the state of the chain. The invariance of the groundstate precisely corresponds to it being an eigenstate of all the symmetry operators on the pair of decoupled spins.

The particular eigenvalues of $\ket{\xi}$ given in Table \ref{tab:xieig}, under symmetries
of $W^{AB}$ which generate conserved quantities on the pair of chains, will determine the evolution of the encoded qubits caused by the decoupling process. The total evolution of the encoded qubits caused by the decoupling in Eq. \eqref{2qbit} turns out to be a CZ gate followed by Pauli $\sigma^{\hat{x}}$ operators on each qubit, which is a nontrivial entangling gate.

To calculate the evolution we first consider $\pi$-rotations about each $\hat{z}$ axis $(R^{\hat{z}},1)$ and $(1,R^{\hat{z}})$. For an initial state in the combined $S^{\hat{z}}$ product basis, $\ket{\epsilon_1}^A_n\ket{\epsilon_2}^B_n$ where $\epsilon_1,\epsilon_2\in \{ 0, 1 \}$, we have
\begin{align}
\ket{\Psi^{\epsilon_1 \epsilon_2}(0)}&= \left| \Sigma_{n}^{\hat{z}}\otimes1 =(-1)^{\epsilon_1},1\otimes\Sigma_{n}^{\hat{z}}=(-1)^{\epsilon_2},H_{1,n}^{A,\text{Haldane}}=0,H_{1,n}^{B,\text{Haldane}}=0\right> \\
&= \ket{\epsilon_1}^A_n\ket{\epsilon_2}^B_n \nonumber
\end{align}
which becomes, after the decoupling
\begin{align}
\label{eq:xizz}
\ket{\Psi^{\epsilon_1' \epsilon_2'}(T)}&= \left| W^{AB}=0,\Sigma_{n}^{\hat{z}}\otimes1=(-1)^{\epsilon_1},1\otimes\Sigma_{n}^{\hat{z}}=(-1)^{\epsilon_2},H_{2,n}^{A,\text{Haldane}}=0,H_{2,n}^{B,\text{Haldane}}=0\right> \nonumber \\
&=\ket{\xi}\otimes \ket{\Sigma_{n}^{\hat{z}}={-}({-}1)^{\epsilon_1},H_{2,n}^{A,\text{Haldane}}=0} \otimes \ket{\Sigma_{n}^{\hat{z}}={-}({-}1)^{\epsilon_2},H_{2,n}^{B,\text{Haldane}}=0} \nonumber \\
&= \ket{\xi} \otimes \ket{\epsilon_1{+}1}^A_{n-1} \otimes \ket{\epsilon_2 {+}1}^B_{n-1}
\end{align}
up to an unknown phase $\theta_{\epsilon_1 \epsilon_2}$, where the addition inside these kets is mod 2. To determine the state in Eq. \eqref{eq:xizz} we have used the $-1$ eigenvalue of $\ket{\xi}$ under the $(R^{\hat{z}},1)$ and $(1,R^{\hat{z}})$ rotations on the decoupled spins.
\begin{table}[t]
\centering
\begin{tabular}{| c | c | c |}
\hline
$W^{AB}$ Symmetry & $|\xi\rangle$ Eigenvalue & Corresponding Conserved Quantity \\
\hline
$(\sqrt{R^{\hat{z}}},R^{\hat{x}})$ & $i$ & $\sqrt{\Sigma_{n}^{\hat{z}}}\otimes\Sigma_{n}^{\hat{x}}$ \\
$(R^{\hat{u}},R^{\hat{u}})$ & $1$ & $\Sigma_{n}^{\hat{u}}\otimes\Sigma_{n}^{\hat{u}}$ \\
$(R^{\hat{v}},R^{\hat{v}})$ & $1$ & $\Sigma_{n}^{\hat{v}}\otimes\Sigma_{n}^{\hat{v}}$ \\
$(R^{\hat{z}},1)$ & $-1$ & $\Sigma_{n}^{\hat{z}}\otimes1$ \\
$(1,R^{\hat{z}})$ & $-1$ & $1\otimes\Sigma_{n}^{\hat{z}}$ \\
\hline
\end{tabular}
\caption{The eigenvalues of $|\xi\rangle$ for various different symmetry operators.}
\label{tab:xieig}
\end{table}
Hence the logical evolution must take the form
\begin{equation}
\label{eq:uab2}
 U^{AB}=\begin{bmatrix}
  0 & 0 & 0 & \theta_{11} \\
  0 & 0 & \theta_{10} & 0 \\
  0 & \theta_{01} & 0 & 0 \\
  \theta_{00} & 0 & 0 & 0 
 \end{bmatrix}
\end{equation}
for the unknown phases $\theta_{11}, \theta_{10}, \theta_{01}, \theta_{00}$ defined above. 

To specify the constants $\theta_{\epsilon_1 \epsilon_2}$ we consider the reducible projective representations of the conserved quantities listed in Table \ref{tab:xieig}, and match their eigenvectors to the respective eigenvalues of the irreducible projective representations shown in Table \ref{tab:eigs}. This corresponds to identifying the logical states encoded in the degenerate ground space by the operators in Eq. \eqref{eq:consquant} with the state of the edge mode. The action of the symmetries within the degenerate ground space is then described by the irreducible projective representation on the boundary mode. 

We define a set of states in the degenerate ground space of the two chains
\begin{equation}
\ket{\epsilon_u,\epsilon_v}_n=\left| \Sigma_{n}^{\hat{u}}\otimes\Sigma_{n}^{\hat{u}}=(-1)^{\epsilon_u},\Sigma_{n}^{\hat{v}}\otimes\Sigma_{n}^{\hat{v}}=(-1)^{\epsilon_v},H^{A}_{n}=0,H^{B}_{n}=0\right>\,.
\end{equation}
The $+1$ eigenvalue of the edge state $\ket{\xi}$ under the rotations $(R^{\hat{u}},R^{\hat{u}})$ and $(R^{\hat{v}},R^{\hat{v}})$ specify the evolution of the spin chain initialized in this state to be
\begin{equation}
\ket{\epsilon_u,\epsilon_v}_n\rightarrow\ket{\xi}\otimes\ket{\epsilon_u,\epsilon_v}_{n-1}
\end{equation}
and so we have
\begin{align*}
&\ket{\epsilon_u=1,\epsilon_v=1} \mapsto \ket{1}\ket{0}+\ket{0}\ket{1} \\
&\ket{\epsilon_u=1,\epsilon_v=0} \mapsto \ket{0}\ket{0} +i \ket{1}\ket{1} \\
&\ket{\epsilon_u=0,\epsilon_v=1} \mapsto \ket{0}\ket{0} -i \ket{1}\ket{1} \\
&\ket{\epsilon_u=0,\epsilon_v=0} \mapsto \ket{1}\ket{0}-\ket{0}\ket{1} 
\end{align*}
To this end we will compare the action of $U^{AB}$ on the states encoded at the boundary to the effect of the adiabatic evolution on the full spin chains, allowing us to determine the unknown constants.

The adiabatic evolution takes the set of initial states
\begin{equation} 
\ket{\Psi^{\epsilon_u\epsilon_v}(0)}=\left| \Sigma_{n}^{\hat{u}}\otimes\Sigma_{n}^{\hat{u}}=(-1)^{\epsilon_u},\Sigma_{n}^{\hat{v}}\otimes\Sigma_{n}^{\hat{v}}=(-1)^{\epsilon_v},H^{A}_{n}=0,H^{B}_{n}=0 \right>
\end{equation}
(which are possible since the two operators $\Sigma_{n}^{\hat{u}}\otimes\Sigma_{n}^{\hat{u}}$ and ${\Sigma_{n}^{\hat{v}}\otimes\Sigma_{n}^{\hat{v}}}$ commute)
to the final states
\begin{equation} 
\ket{\Psi^{\epsilon_u\epsilon_v}(T)}=\ket{\xi}\otimes \left|\Sigma_{n}^{\hat{u}}\otimes\Sigma_{n}^{\hat{u}}=(-1)^{\epsilon_u}, \Sigma_{n}^{\hat{v}}\otimes\Sigma_{n}^{\hat{v}}=(-1)^{\epsilon_v},H^{A}_{n-1}=0,H^{B}_{n-1}=0\right>
\end{equation}
since $|\xi\rangle$ has the eigenvalue 1 under the rotations $(R^{\hat{u}},R^{\hat{u}})$ and $(R^{\hat{v}},R^{\hat{v}})$ on the decoupled spins, effectively fixing the $\ket{\epsilon_u,\epsilon_v}_L$ logical states up to a set of phase shifts $\phi_{\epsilon_u\epsilon_v}$.

We focus our attention on the (unnormalized) state $(\ket{1}\ket{0}+\ket{0}\ket{1})$, the joint $-1$ eigenstate of of $i\sigma^{\hat{u}}\otimes i\sigma^{\hat{u}}$ and $i\sigma^{\hat{v}}\otimes i\sigma^{\hat{v}}$, we have
\begin{equation}
U^{AB}\left(\ket{1}\ket{0}+\ket{0}\ket{1}\right) = \theta_{10}\ket{0}\ket{1} + \theta_{01}\ket{1}\ket{0}
\end{equation}
which must agree with the evolution of $\ket{\epsilon_u=1,\epsilon_v=1}$ that merely accumulates a phase shift $\phi_{11}$. Hence after the evolution we have: $\phi_{11}(\ket{1}\ket{0}+\ket{0}\ket{1})=\left( \theta_{10}\ket{0}\ket{1} + \theta_{01}\ket{1}\ket{0} \right)$, which requires that $\theta_{10}=\theta_{01}$.

Similarly we consider the evolution of $(\ket{0}\ket{0}+i\ket{1}\ket{1})$,  the -1 eigenstate of of $i\sigma^{\hat{u}}\otimes i\sigma^{\hat{u}}$ and +1 eigenstate of $i\sigma^{\hat{v}}\otimes i\sigma^{\hat{v}}$
\begin{equation}
U^{AB}\left(\ket{0}\ket{0}+i\ket{1}\ket{1}\right) = \theta_{00}\ket{1}\ket{1} +i\theta_{11}\ket{0}\ket{0}
\end{equation}
which must agree with the evolution of the state $\ket{\epsilon_u=1,\epsilon_v=0}_L$ encoded in the spin chains. Hence the equality: $\phi_{11}\left(\ket{0}\ket{0}+i\ket{1}\ket{1}\right) = \left( \theta_{00}\ket{1}\ket{1} +i\theta_{11}\ket{0}\ket{0} \right)$, which specifies $\theta_{00}=-\theta_{11}$. 

\begin{table*}[t]
\centering
\begin{tabular}{| c | c | c |}
\hline
Conserved quantity & Projective representation & eigenvectors grouped by eigenvalue \\
\hline
$\Sigma_{n}^{\hat{u}}\otimes\Sigma_{n}^{\hat{u}}$ & $i\sigma^{\hat{u}}\otimes i\sigma^{\hat{u}}$ & $ \underbrace{|0\rangle|0\rangle+i|1\rangle|1\rangle,\ |1\rangle|0\rangle+|0\rangle|1\rangle}_{\text{eigenvalue: }-1},\ \underbrace{ |0\rangle|0\rangle-i|1\rangle|1\rangle,\ |1\rangle|0\rangle-|0\rangle|1\rangle}_{+1}$ \\
$\Sigma_{n}^{\hat{v}}\otimes\Sigma_{n}^{\hat{v}}$ & $i\sigma^{\hat{v}}\otimes i\sigma^{\hat{v}}$ & $  \underbrace{|0\rangle|0\rangle-i|1\rangle|1\rangle,\ |1\rangle|0\rangle+|0\rangle|1\rangle}_{-1},\ \underbrace{ |0\rangle|0\rangle+i|1\rangle|1\rangle,\ |1\rangle|0\rangle-|0\rangle|1\rangle}_{+1}$  \\
$\sqrt{\Sigma_{n}^{\hat{z}}}\otimes\Sigma_{n}^{\hat{x}}$ & $\frac{1}{\sqrt{2}} (I+i\sigma^{\hat{z}})\otimes i\sigma^{\hat{x}}$ & $\underbrace{|1\rangle|1\rangle-|1\rangle|0\rangle}_{e^{-i\frac{3\pi}{4}}},\ \underbrace{|1\rangle|1\rangle+|1\rangle|0\rangle}_{e^{i\frac{\pi}{4}}},\ \underbrace{|0\rangle|1\rangle-|0\rangle|0\rangle}_{e^{-i\frac{\pi}{4}}},\ \underbrace{|0\rangle|1\rangle+|0\rangle|0\rangle}_{e^{i\frac{3\pi}{4}}} $ \\
\hline
\end{tabular}
\caption{The eigenvectors of the projective representations of various symmetry operators.}
\label{tab:eigs}
\end{table*}

Finally we consider the $e^{-i\frac{3\pi}{4}}$ eigenstate of $\left[ \frac{1}{\sqrt{2}} (I+i\sigma^{\hat{z}})\otimes i\sigma^{\hat{x}} \right]$, $(\ket{1}\ket{1}-\ket{1}\ket{0})$ which evolves to
\begin{equation}
U^{AB}\left(\ket{1}\ket{1}-\ket{1}\ket{0}\right) = \theta_{11}\ket{0}\ket{0} - \theta_{10}\ket{0}\ket{1}
\end{equation}
and the corresponding encoded state
\begin{equation}
\ket{\Psi(0)}= \left| \sqrt{\Sigma_{n}^{\hat{z}}}\otimes\Sigma_{n}^{\hat{x}}=e^{-i\frac{3\pi}{4}},H_{1,n}^{A,\text{Haldane}}=0,H_{1,n}^{B,\text{Haldane}}=0\right> 
\end{equation}
evolves to
\begin{equation}
\ket{\Psi(T)}= \ket{\xi} \otimes \left| \sqrt{\Sigma_{n-1}^{\hat{z}}}\otimes\Sigma_{n-1}^{\hat{x}}=e^{i\frac{3\pi}{4}},H_{2,n}^{A,\text{Haldane}}=0,H_{2,n}^{B,\text{Haldane}}=0\right>
\end{equation}
up to a multiplicative phase $\phi$, due to the conservation of $\sqrt{\Sigma_{n}^{\hat{u}}}\otimes\Sigma_{n}^{\hat{x}}$ and the $i$-eigenvalue of $\ket{\xi}$ under the symmetry operator $(\sqrt{R^{\hat{z}}},R^{\hat{x}})$ on the decoupled spins. For these to agree requires that: $(\theta_{11}\ket{0}\ket{0} - \theta_{10}\ket{0}\ket{1})=\phi(\ket{0}\ket{1}+\ket{0}\ket{0})$ and hence $\theta_{11}=-\theta_{10}$ which specifies $U^{AB}$ up to an irrelevant global phase.

\[
 U^{AB}=\begin{bmatrix}
  0 & 0 & 0 & {-}1 \\
  0 & 0 & 1 & 0 \\
  0 & 1 & 0 & 0 \\
  1 & 0 & 0 & 0 
 \end{bmatrix}= (\sigma^{\hat{x}}\otimes \sigma^{\hat{x}}) \text{CZ}
\]
which constitutes a CZ gate followed by a simultaneous Pauli $\sigma^{\hat{x}}$ operator on each of the encoded qubits. To complete this holonomy we could consider undoing the $\sigma^{\hat{x}}$ operators on each chain using the reverse of the evolution described in Eq. \eqref{eq:szhol}.

We have seen, in this section and the previous one, how adiabatic holnomic evolutions of the spin chains can cause unitary logical evolutions of the qubits encoded within their degenerate ground states. In the next section we will look more closely at the symmetry group of the two chain interaction $G_2$.


\subsubsection{Symmetry Group of the Two-Chain Interaction} \label{sec:symgp}

In this section, we will examine the structure of the symmetry group $G_2$ of the two-qubit coupling Hamiltonian $W^{AB}$ in detail. This group is important since it will determine whether or not the two-qubit gate is symmetry protected. We will determine the full set of elements within this group and use this description to identify it with one of the isomorphism classes of all size 16 groups, specifically the class labeled by $D_2\rtimes\mathbb{Z}_4$. 

The symmetries of $W^{AB}$ that were explicitly used in the calculation of the two-qubit gate are given in Table \ref{tab:xieig}. From these we can see that the symmetry group $G_2$ consists of a discrete set of joint rotations of each pair of spins from the two chains. We have found that the set of symmetries listed in Table \ref{tab:xieig} are not independent, and the group can be generated by the rotations
$(R^{\hat{u}},R^{\hat{u}})$, $(R^{\hat{v}},R^{\hat{v}})$, and $(\sqrt{R^{\hat{z}}},R^{\hat{x}})$,
where $\hat{u}=\frac{1}{\sqrt{2}}(\hat{x}+\hat{y})$ and $\hat{v}=\frac{1}{\sqrt{2}}(\hat{x}-\hat{y})$.
We can decompose the rotation $R^{\hat{x}}= \sqrt{R^{\hat{z}}}\ \cdot\ R^{\hat{u}}$ which leads us to conclude that the element $(R^{\hat{v}},R^{\hat{v}})$ can be written as a product of the other two. Hence the group has only two independent generators, depicted in Figure \ref{fig:2sym2}. We have written out these three redundant generators since they allow us to more easily identify this rotation group with the semidirect product group $D_2 \rtimes \mathbb{Z}_4$.

\begin{figure}[ht]
\centering
\includegraphics[width=8cm]{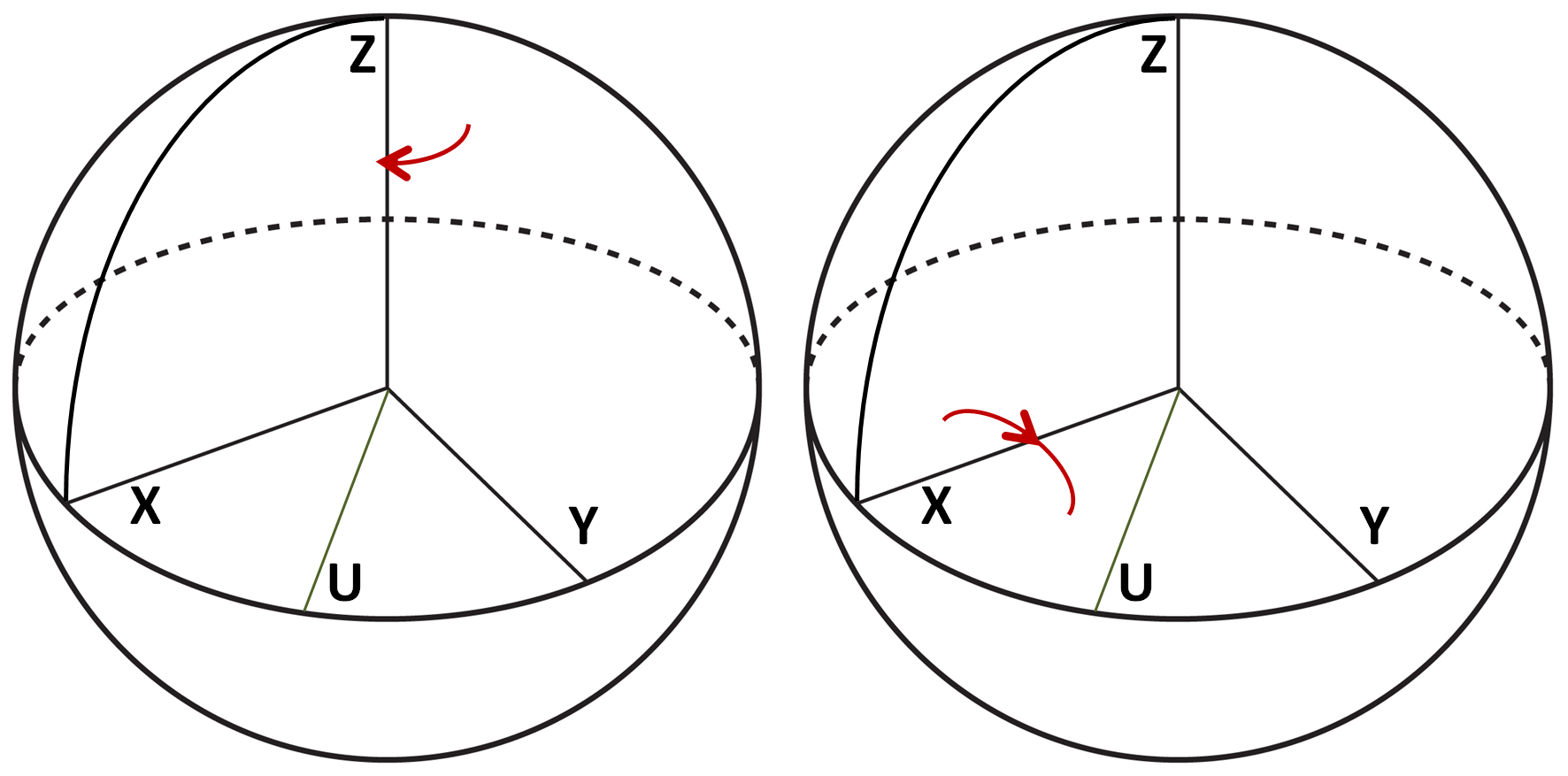} \qquad
\includegraphics[width=8cm]{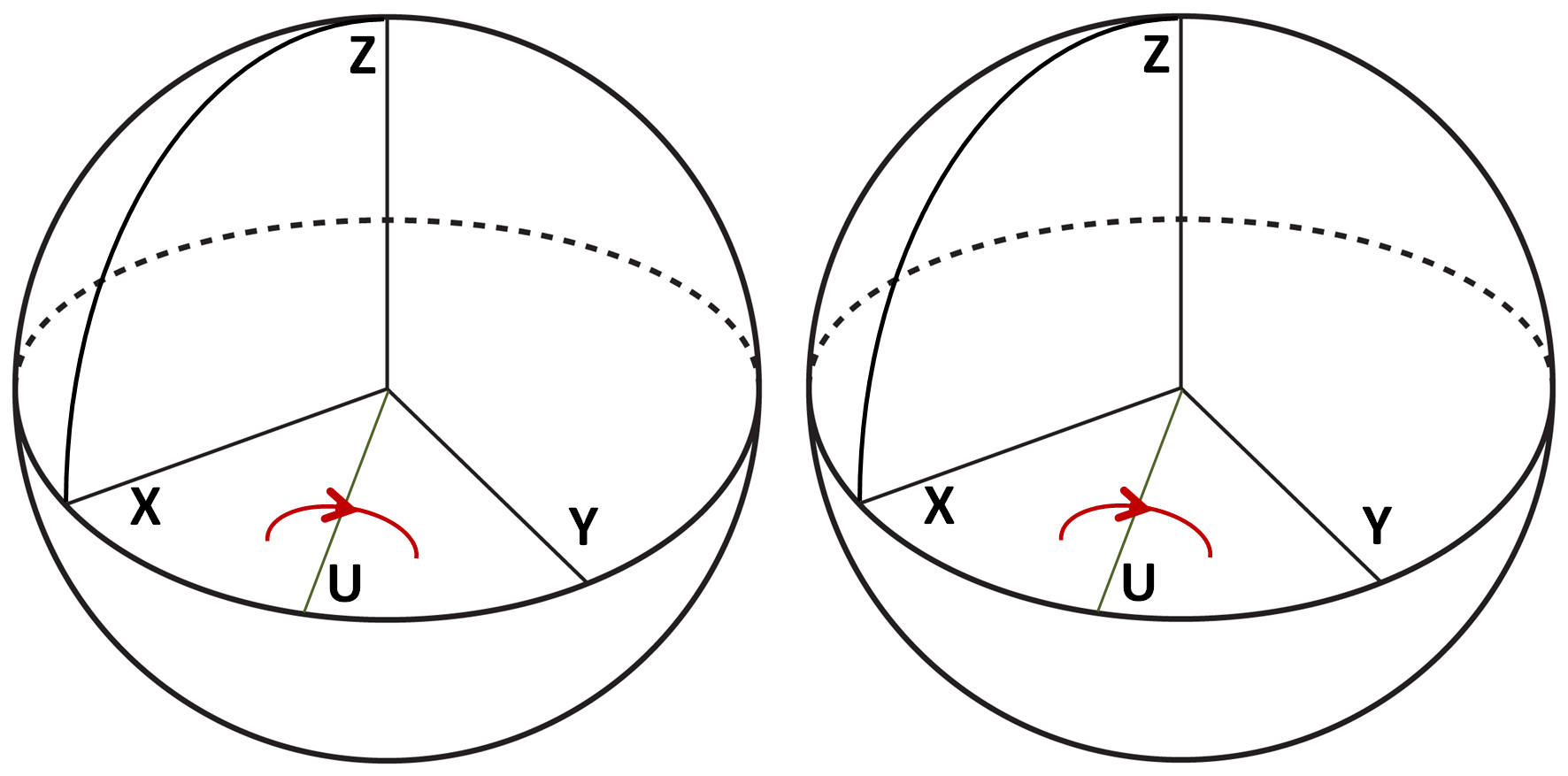}
\caption{The $(\sqrt{R^{\hat{z}}},R^{\hat{x}})$ rotation (left) and the $ (R^{\hat{u}},R^{\hat{u}})$ rotation (right) generating the symmetry group.}
\label{fig:2sym2}
\end{figure}

As a set, $D_2\rtimes\mathbb{Z}_4$ is made up of the same elements as the direct product $ D_2  \times \mathbb{Z}_4$  but possesses a different multiplication rule and hence a different group structure, given by
\begin{equation}
(n_1,h_1)\times (n_2,h_2)= \left(n_1 (h_1 n_2 h_1^{-1} ), h_1h_2\right), \qquad
n_1,n_2 \in  D_2 , \quad h_1,h_2\in \mathbb{Z}_4 
\end{equation}
Under this multiplication rule, the semidirect product group is non-Abelian even though it is built from Abelian components. 

To see how the group $G_2$ has the structure of $D_2\rtimes \mathbb{Z}_4$, we note that the subgroup generated by $(\sqrt{R^{\hat{z}}},R^{\hat{x}})$ is isomorphic to $\mathbb{Z}_4$, and that generated by $(R^{\hat{u}},R^{\hat{u}})$ is isomorphic to $\mathbb{Z}_2$. It is the interaction of these two terms via multiplication that produces $(R^{\hat{v}},R^{\hat{v}})$ which, along with $(R^{\hat{u}},R^{\hat{u}})$, generates $D_2$. It is the non-commutativity of the rotations that generates the non-Abelian structure of the semidirect product $D_2\rtimes \mathbb{Z}_4$ upon combining the two subgroups.

A useful description of $D_2\rtimes\mathbb{Z}_4$ is the presentation in terms of its generators and their relations
\begin{equation}
\label{eq:genrel}
\left[ \alpha, \beta | \alpha^4=\beta^2=1 , \alpha\beta=(\alpha \beta)^2 \beta \alpha^3 \right]
\end{equation}
this notation means the set of all products of $\alpha$ and $\beta$ where these two elements satisfy the relations on the right of Eq. \eqref{eq:genrel}. Using this description we have identified $\alpha \mapsto (\sqrt{R^{\hat{z}}},R^{\hat{x}})$ and $\beta \mapsto (R^{\hat{u}},R^{\hat{u}})$, which specifies the correspondence for all the other group elements and establishes the isomorphism $G_2 \cong D_2\rtimes\mathbb{Z}_4$.

To make this identification, we made use of the derived subgroup $G_2'=[G_2,G_2]:=\left\{ g_1 g_2 g_1^{-1} g_2^{-1} \middle| g_1,g_2 \in G_2\right\}$.  Note $G_2'$ is nontrivial in this case, as a consequence of the group's non-Abelian structure, and provides us with a useful tool as it will allow us to uniquely determine the correspondence between a nontrivial element from each of the different descriptions.  This is because it consists of only two elements, given by $G_2'= \{1, (\alpha \beta)^2 \alpha^2 \} \cong  \{ 1,  (R^{\hat{z}},R^{\hat{z}})\}$ allowing us to identify $(\alpha \beta)^2 \alpha^2 \cong(R^{\hat{z}},R^{\hat{z}})$ and pin down the exact structure.

Another useful tool is the center subgroup $Z(G_2)$, consisting of elements which commute with the whole group. For $G_2$ it is given by $Z(G_2)=\{1, \alpha^2, (\alpha\beta)^2,(\alpha\beta)^2\alpha^2 \} \cong \{ 1, (R^{\hat{z}},1),(1,R^{\hat{z}}), (R^{\hat{z}},R^{\hat{z}})\}$, i.e., with the structure $\mathbb{Z}_2 \times \mathbb{Z}_2$. We will harness the commutative property of this subgroup in the following to prove that a particular projective representation of $G_2$ is nontrivial.

\end{widetext}

\subsubsection{The Symmetry Group Protects a nontrivial SP Phase}\label{sec:pfnontriv}\label{sec:repsofG2}

We will now show that $G_2 \cong D_2\rtimes\mathbb{Z}_4$ has nontrivial projective representations, and therefore protects a nontrivial SP phase.  We will prove this by considering its action on a pair of spin-$1/2$ particles, and explicitly demonstrate its nontriviality.  To be precise, the symmetry group $G_2\cong D_2\rtimes\mathbb{Z}_4$ is an embedding of the semidirect product group $D_2\rtimes\mathbb{Z}_4$ into two copies of the 3D rotation group  $G_2 \subset SO(3)\times SO(3)$. We have used this identification to construct the canonical linear representation of $G_2$ on a pair of spin-1 particles by mapping each rotation to its corresponding spin-1 rotation operator.
Because the symmetry group $G_2$ is generated by the two rotations, $(\sqrt{R^{\hat{z}}},R^{\hat{x}})$ and $(R^{\hat{u}},R^{\hat{u}})$, we need only define the operators that represent each of these rotations to uniquely specify a representation. These are given by
\begin{align}
(\sqrt{R^{\hat{z}}},R^{\hat{x}})&\mapsto \exp (i\frac{\pi}{2} S^{\hat{z}}) \otimes \exp (i\pi S^{\hat{x}}), \nonumber \\ (R^{\hat{u}},R^{\hat{u}}) &\mapsto \exp (i\pi S^{\hat{u}}) \otimes \exp (i\pi S^{\hat{u}}) 
\end{align}
which act upon pairs of spins from the two chains respectively.

In a similar way, we can construct \emph{the $\half \otimes \half$ projective representation of $G_2$} on a pair of spin-$\half$ particles by mapping each rotation to its corresponding spin-$\half$ rotation operator
\begin{align}\label{eq:g2repdefn}
(\sqrt{R^{\hat{z}}},R^{\hat{x}})&\mapsto \exp (i\frac{\pi}{4} \sigma^{\hat{z}}) \otimes \exp (i\frac{\pi}{2} \sigma^{\hat{x}}), \nonumber \\
 (R^{\hat{u}},R^{\hat{u}}) &\mapsto \exp (i\frac{\pi}{2} \sigma^{\hat{u}}) \otimes \exp (i\frac{\pi}{2} \sigma^{\hat{u}}). 
\end{align}
We will show that this projective representation is irreducible and nontrivial.

We introduced the classification of 1D SP phases protected by an on-site symmetry group $G$ in terms of its second cohomology class $H^2(G,U(1))$ in Sec.~\ref{sec:SPintro}. Each ground state is labeled by the equivalence class $[\omega]$ of the factor set $\omega$ induced by the projective representation of $G$ acting upon its boundary modes.  We will show now that the symmetry group $G_2$ of the coupling Hamiltonian $W^{AB}$ protects a nontrivial SP phase with four-dimensional $\half\otimes\half$ boundary excitations and hence the two-qubit gate is symmetry-protected in the same sense as the single-qubit gates. Before doing so, we will introduce the arguments used in the proof, with the known example of a single chain protected by a $D_2$ symmetry group.

It is well known that the $D_2$ symmetry protects a nontrivial SP phase, labeled by the nontrivial element of its second cohomology class $H^2(D_2,U(1))=\mathbb{Z}_2$. This phase is characterized by the irreducible, nontrivial spin-$\half$ representation of $D_2$ that maps the non-unit group elements to Pauli operators $X$, $Y$ and $Z$ (up to multiplicative constants). 

We will go through the arguments leading to the conclusion that this representation of $D_2$ is nontrivial and irreducible, developing tools that will be useful in making the same arguments for $G_2$. To prove the irreducibility condition we will first outline a lemma which is key to the type of arguments we want to make.
\begin{lem}
\label{lem:1}
Any proper, nontrivial invariant subspace $\mathcal{W}\subsetneq \mathcal{H}$ of a unitary matrix $U:\mathcal{H}\rightarrow \mathcal{H}$ is spanned by the eigenvectors of the matrix $U\restr_{\mathcal{W}}$ formed by restricting $U$ to the subspace $\mathcal{W}$.
\end{lem}
\noindent For a proof, see Ref.~\cite{radjavi2003invariant}.  (We make explicit mention of this fact because it is not true for general matrices. For a general matrix $M$ assuming that all invariant subspaces are spanned by eigenvectors of $M$ is known as the `diagonal fallacy'.)

We can easily see that the Pauli projective representation of $D_2$ is irreducible, because all its matrices are unitary and no two Pauli operators share a joint eigenspace. Hence by Lemma \ref{lem:1} there can be no proper subspace left invariant under all actions in this projective representation.
 
To see that the factor system is nontrivial we must show that multiplication by any 2-coboundary cannot take it to the trivial factor system. This is not particularly easy to see directly and so we introduce a function $\varphi$ that will give us an easily calculable, sufficient condition for the nontriviality of a particular factor system. This function $\varphi:G\rightarrow U(1)$ is given (see Ref.~\cite{costache2009irreducible}) by the sum
\begin{equation}
\label{factfunc}
\varphi_{\omega} (a) = \sum_{b\in G} \left( \frac{ \omega(a,b)}{ \omega (b,a) } \right) = \sum_{b\in G} { \omega(a,b)} { \omega^* (b,a) } 
\end{equation}
since $\omega^{-1}=\omega^*$ for a complex phase $\omega \in U(1)$.
\begin{lem}
\label{lem:2}
For any two factor systems of $G$, say $\omega$ and $\nu$, the existence of some $a \in Z(G)$ for which $\varphi_{\omega}(a)\neq \varphi_{\nu}(a)$ implies that $[\omega]\neq [\nu]$ (i.e. that $\omega$ and $\nu$ lie in distinct cohomology classes).
\end{lem}
\vspace{-0.3cm}
\begin{proof}
Suppose we are given two factor systems $\omega$ and $\nu$ from the same cohomology class, by definition they are related by some 2-coboundary function $\left[\mu(ab)/\mu(a)\mu(b)\right]$ i.e.
\begin{equation}
\label{eq:2cobp}
\mu(ab)\omega(a,b)=\mu(a)\mu(b)\nu(a,b)\ .
\end{equation}
This implies that the functions $\varphi_{\omega}$ and $\varphi_{\nu}$ are equal for all group elements in the center\footnote{For an Abelian group $G=Z(G)$, hence $\omega\sim\nu\implies\varphi_{\omega}\equiv\varphi_{\nu}$.} of the symmetry group, i.e. $ a\in Z(G)\implies \varphi_{\omega}(a)=\varphi_{\nu}(a)$,  since
\begin{align}
\varphi_{\omega} (a)&= \sum_{b\in G} \left( \frac{ \omega(a,b)}{ \omega (b,a) } \right)  = \sum_{b\in G} \left( \frac{ \mu(a)\mu(b)}{ \mu(b)\mu(a) } \frac{\mu(ba) }{\mu(ab) } \frac{\nu(a,b)}{ \nu(b,a) } \right) \nonumber \\
&=  \sum_{b\in G} \left(  \frac{\nu(a,b)}{ \nu(b,a) } \right) = \varphi_{\nu}(a)
\end{align}
where we have used Eq.~\eqref{eq:2cobp} and then the fact that $ab=ba$, $\forall\  b\in G,\ \forall\  a\in Z(G)$.
So far we have shown that $[\omega]=[\nu]\implies[ a\in Z(G)\implies \varphi_{\omega}(a)=\varphi_{\nu}(a)]$, taking the contrapositive of this statement we achieve the desired result
\end{proof}
We will now use Lemma~\ref{lem:2} to give a simple sufficient condition for the nontriviality of a factor system. We first note that for the trivial factor system (defined to be $1 (a,b)=1$, $\forall a,b\in G$) we have $\varphi_{1} (a) = |G|$ for all group elements $a\in G$. Hence for any factor system $\omega$, the existence of a group element $a \in Z(G)$ for which $\varphi_\omega (a) \neq |G|$ implies by Lemma~\ref{lem:2} that $\omega$ is nontrivial.

\begin{table}[t]
\centering
\begin{tabular}{| c || c | c | c | c |}
\hline
$\omega$ & $I$ & $X$ & $Y$ & $Z$ \\
\hline \hline
$I$ & $1$ & $1$ & $1$ & $1$ \\
\hline
$X$ & $1$ & $1$ & $i$ & $-i$ \\
\hline
$Y$ & $1$ & $-i$ & $1$ & $i$ \\
\hline
$Z$ & $1$ & $i$ & $-i$ & $1$ \\
\hline
\end{tabular}
\caption{The factor system of the Pauli projective representation of $D_2$.}
\label{tab:fs}
\end{table}

In the particular case of the Pauli projective representation of $D_2$ one can readily verify from the factor system in Table~\ref{tab:fs} that $\varphi_\omega(X)=\varphi_\omega(Y)=\varphi_\omega(Z)=0$ and hence that the factor system $\omega$ is nontrivial.

We now show that the particular projective representation of $G_2$ given by qubit rotation operators on the four-dimensional $\half \otimes \half$ Hilbert space is irreducible and nontrivial using the arguments introduced above. 
\begin{thm}
\label{thm:1}
The projective representation of $G_2$, constructed by mapping each pair of rotations to their corresponding rotation operators on the four dimensional $\half \otimes \half$ Hilbert space, is irreducible and nontrivial.
\end{thm}

\begin{proof}
To prove the irreducibility of the $\half \otimes \half$ projective representation of the symmetry group $G_2$ we consider the possibilities for invariant subspaces of $\mathbb{C}^4$ under all the $4\times4$ matrices in this representation of $G_2$. A consequence of Lemma \ref{lem:1} is that any unitary matrix $U:\mathbb{C}^n \rightarrow \mathbb{C}^n$ which possesses an invariant subspace $\mathcal{W}$ must have a block structure, also leaving the orthogonal subspace $\mathcal{W}^{\perp}$ invariant under $U$. 

For the $\half \otimes \half$ projective representation of $G_2$, we are left with the possibility of either two invariant two-dimensional subspaces, or a one- and a three-dimensional invariant subspace. (Note that the invariant subspaces of dimension $d >1$ may split up further.)  We can immediately exclude the later case, because the existence of a one-dimensional invariant subspace would require all matrices in the projective representation to share an eigenstate due to Lemma~\ref{lem:1}, and an explicit calculation (not shown) confirms that this is not the case.

Any two-dimensional invariant subspace of a matrix $U$ in the projective representation must be the span of two eigenvectors, as guaranteed by Lemma~\ref{lem:1}. If all matrices of the projective representation are to share such an invariant subspace $\mathcal{W}$ it must be possible to write at least one eigenvector of any matrix $U_1$ as a linear combination of two eigenvectors of any other matrix $U_2$. Again we have confirmed that this is not the case for the $\half \otimes \half$  projective representation of $G_2$. A counter example is given by considering the eigenvectors of the matrices representing the group elements $(\sqrt{R^{\hat{z}}},R^{\hat{x}})$ and $(R^{\hat{x}},\sqrt{R^{\hat{z}}})$. In particular, we can see by explicitly listing the eigenvectors of the matrix representation \eqref{eq:g2repdefn} that no linear combination of any two eigenvectors of one matrix can be used to form an eigenvector of the other. Hence the $\half \otimes \half$ projective representation must be irreducible, since there cannot be any subspace invariant under all of its operators.

To show that the $\half \otimes \half$ projective representation has a nontrivial factor system $\omega$ we consider the function $\varphi_\omega$ defined in Eq. \eqref{factfunc}. Then we calculate $\varphi_\omega$ directly from the factor system (implicitly defined by the choice of representation) to find $\varphi_{\omega} (R^{\hat{z}},1)=\varphi_{\omega} (1,R^{\hat{z}})=\varphi_{\omega} (R^{\hat{z}},R^{\hat{z}})=0$, a sufficient condition to conclude that $\omega$ cannot lie within the trivial cohomology class after invoking Lemma \ref{lem:2}. 
\end{proof}

\begin{cor}
The symmetry group $G_2$ protects at least one nontrivial SP phase, labeled by the factor system $\omega$ of the representation described in Eq.~\eqref{eq:g2repdefn}.
\end{cor}

In this section we have shown that the symmetry group $G_2$ of the two-qubit coupling Hamiltonian $W^{AB}$ protects a nontrivial SP phase which supports two-qubit boundary modes. Hence the two-qubit gate generated under this Hamiltonian is protected against symmetric perturbations by $G_2$. In the next section we will use this result, along with the symmetry-protection of the single-qubit operation to find a minimal symmetry requirement for universal quantum computation using only these symmetry-protected gates.

\subsection{Symmetry Requirements for Universal Quantum Computation}\label{sec:symreq}

Using the results of the previous section, we now give minimal symmetry requirements for universal quantum computation with only symmetry-protected gates.

The process used to generate arbitrary single-qubit gates proposed in Ref.~\cite{renes2011holonomic} requires a continuous, time-dependent embedding of the $D_2$ symmetry, that protects the single-qubit gate, within the full $SO(3)$ symmetry of the chain. For practical simplicity (and to avoid such considerations) we will now investigate the minimal set of different symmetries required for universal quantum computation with only symmetry-protected gates. Recall that a universal gate set can be achieved by generating arbitrary single-qubit gates along with a nontrivial entangling gate.

For the single-qubit gates, our argument relies on the geometric result of applying a pair of $\pi$-rotations about non-orthogonal axes depicted in Fig.~\ref{fig:cosrot}. Applying a $\pi$-rotation about the $\hat{m}$ axis, followed by a $\pi$-rotation about the $\hat{m}' $ axis amounts to a total rotation through the angle $\left[ 2\cos^{-1} (\hat{m}\cdot \hat{m}' ) \right]$ about the $\hat{m}\times \hat{m}'$ axis. In this way we can perform a smaller rotation of the encoded qubit by applying the one qubit gate described in Eq. \eqref{1qbitholonomy} twice, picking a different, but fixed, embedding of $D_2$ for each evolution.

To find exactly what symmetry embeddings are sufficient to simulate any single-qubit unitary transformation efficiently, we make use of the Solovay-Kitaev theorem. Specifically we will use a corollary of this theorem described in Ref.~\cite{nielsen2002quantum} which ensures that any single-qubit unitary can be efficiently decomposed into a product of Hadamard, Phase and $\pi/8$ gates. 

Hence we need only generate enough different embeddings of the $D_2$ symmetry to make performing these three gates possible, through the repeated application of $\pi$-rotations. The Hadamard gate can be performed using an embedding with $\pi$-rotations about the axes $\hat{\mu},\hat{\nu},\hat{y}$, where $\mu=x+z$, $\nu=x-z$. The Phase gate requires a combination of two different embeddings, a suitable choice is given by the standard embedding defined by the $\hat{x},\hat{y},\hat{z}$ axes accompanied by a rotation of this embedding by $(-\pi /4)$ about the $\hat{z}$ axis. Similarly the $\pi /8$ gate can be generated with a combination of two embeddings, using the previously chosen standard embedding defined by the $\hat{x},\hat{y},\hat{z}$ axes along with another rotation of this embedding, this time by $(-\pi/8)$ about the $\hat{z}$ axis, will suffice. So we only need four independent embeddings of $D_2\subset SO(3)$, as depicted in Fig.~\ref{fig:embed}, to generate arbitrary single-qubit gates (note that we could in fact use only three embeddings as the Phase gate is generated by two consecutive $\pi/8$ gates).

Hence we have identified the ability to perform universal quantum computation in this model with the ability to generate four independent embeddings of the $D_2$ symmetry protecting the single-qubit gate combined with the ability to generate the nontrivial two-qubit entangling gate.

\begin{figure}[ht]
\centering
\includegraphics[width=5.5cm]{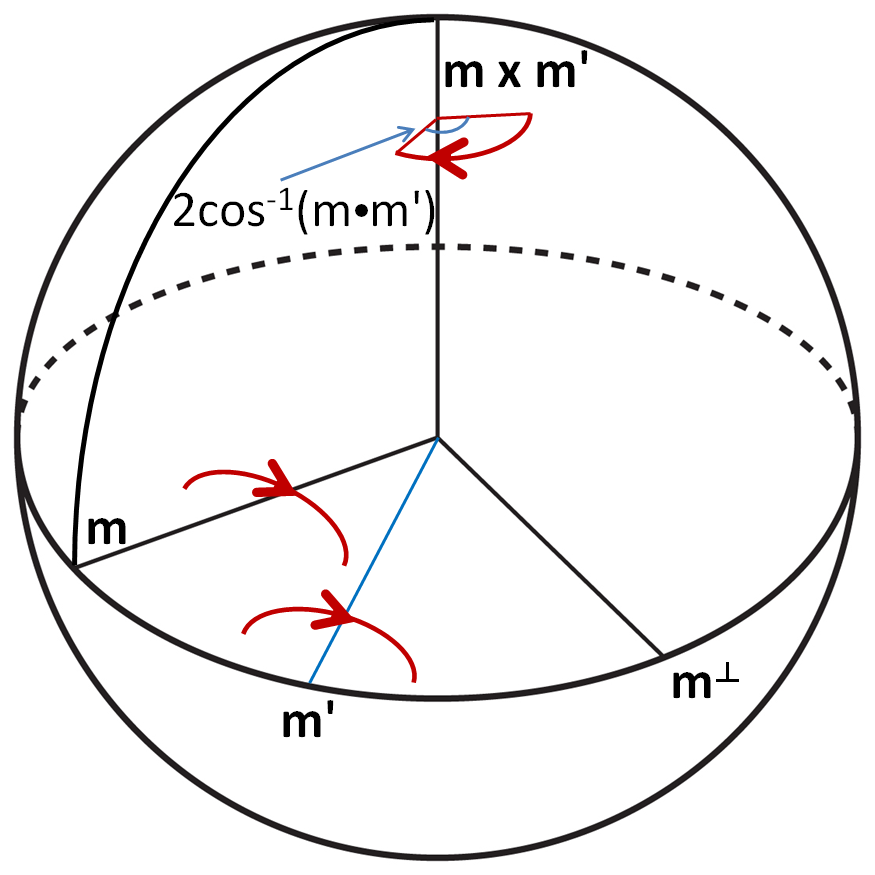}
\caption{Result of consecutive $\pi$-rotations about two non-orthogonal axes.}
\label{fig:cosrot}
\end{figure}

The benefits of finite embedding include ease of implementation, and the ability at each point to tolerate more general perturbations to the Hamiltonian (symmetric under a specific set of $\pi$-rotations). However, since we produce a set of gates dense in all rotations of the encoded qubits, the set of perturbations to the Hamiltonian that are protected against throughout all the single-qubit gates must have a rotation symmetry that is dense in all rotations, essentially the same as the full SO(3) symmetry.

\begin{figure}[ht]
\centering
\includegraphics[width=5.5cm]{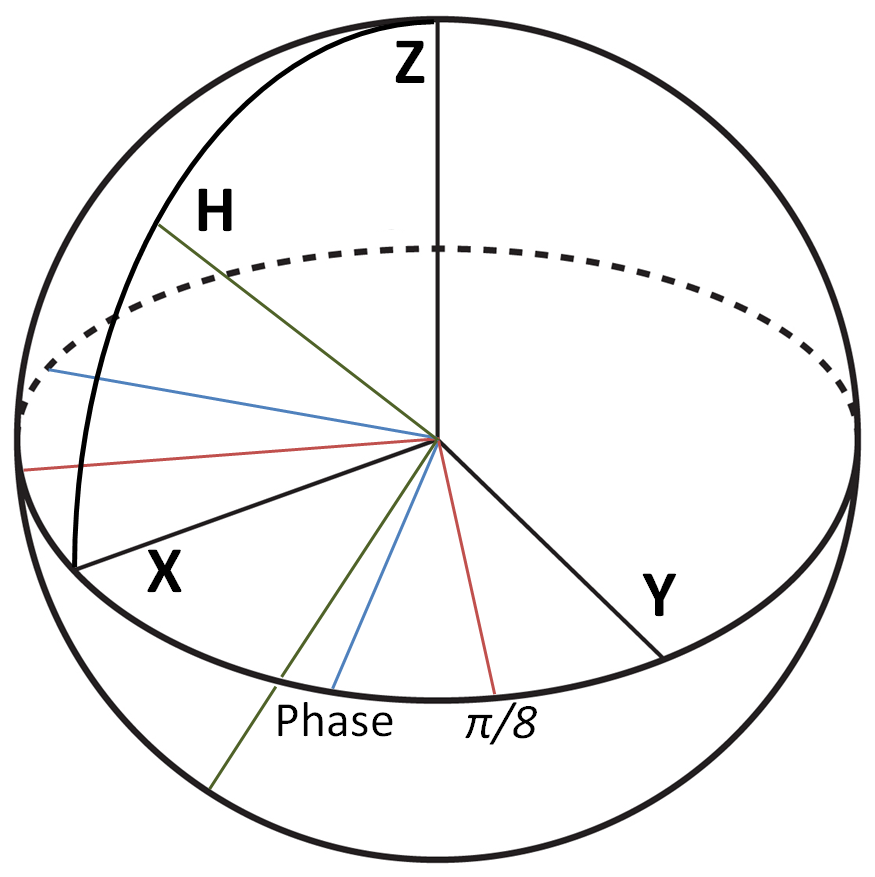}
\caption{Finite symmetry embedding requirement for universal quantum computation.}
\label{fig:embed}
\end{figure}

\end{document}